\documentclass[envcountsame, envcountsect, runningheads, a4paper]{llncs}

\usepackage[dvipsnames]{xcolor}
\usepackage{amsmath, amsfonts,amssymb}

\usepackage{amsthm}
\usepackage[utf8]{inputenc}
\usepackage{graphicx}
\usepackage{tikz}
\usepackage{hyperref}
\usepackage[backgroundcolor=magenta!10!white]{todonotes}
\usepackage{setspace}
\usepackage{enumitem}
\usepackage{ifthen}
\usepackage{caption}
\usepackage{microtype}
\usepackage{stackrel}
\usepackage{proof}
\usepackage[space]{cite}
\usepackage{numprint}
\usepackage{stmaryrd}
\usepackage{cleveref}
\Crefname{section}{Sec.}{Secs.}
\Crefname{proposition}{Prop.}{Props.}
\Crefname{theorem}{Thm.}{Thms.}
\Crefname{lemma}{Lem.}{Lema.}
\Crefname{definition}{Def.}{Defs.}
\Crefname{figure}{Fig.}{Figs.}
\Crefname{table}{Tab.}{Tabs.}
\Crefname{example}{Ex.}{Exps.}
\usepackage{multirow}
\usepackage{subcaption}
\usepackage{verbatim}
\usepackage{placeins}
\usepackage{pgfplots}
\usepgfplotslibrary{fillbetween}
\pgfplotsset{width=3.5cm, compat=1.8, grid style={dashed}}
\usepackage{thmtools}
\usepackage{relsize}
\usepackage{multirow}
\usepackage{multicol}
\usepackage{mathtools}
\usepackage{csquotes}
\usetikzlibrary {arrows , automata , positioning ,calc, arrows.meta, decorations.pathreplacing , intersections, shapes}
\usepackage[dvipsnames]{xcolor}
\usepackage{bm}

\usepackage{academicons}

\usetikzlibrary{arrows,
	automata,
	calc,
	decorations,
	decorations.pathreplacing,
	positioning,
	shapes}

\numberwithin{equation}{section}
\npdecimalsign{.}
\npthousandsep{,}

\colorlet{christelColor}{Apricot!30!white}
\colorlet{tobiasColor}{Cerulean!30!white}
\colorlet{davidColor}{YellowGreen!30!white}
\colorlet{simonColor}{Yellow!30!white}
\colorlet{saschaColor}{Purple!30!white}

\newcommand{\Act}{\operatorname{Act}}

\newcommand{\M}{\mathcal{M}}

\newcommand{\Paths}{\operatorname{Paths}}
\newcommand{\Pathsfin}{\operatorname{Paths}_{\operatorname{fin}}}

\newcommand{\last}{\operatorname{last}}
\newcommand{\lang}{\mathcal{L}}

\newcommand{\prb}{\mathbf{Pr}}
\newcommand{\prism}{\textsc{Prism}}

\newcommand{\breakp}{\mathcal{BP}}
\newcommand{\gfm}{\mathcal{G}^{\operatorname{GFM}}}

\newcommand{\gfmbridge}{\delta_{\operatorname{bridge}}^{\operatorname{GFM}}}
\newcommand{\gfmtrans}{\delta^{\operatorname{GFM}}}
\newcommand{\ld}{\mathcal{G}^{\operatorname{LD}}}

\newcommand{\ldbridge}{\delta_{\operatorname{bridge}}^{\operatorname{LD}}}
\newcommand{\ldtrans}{\delta^{\operatorname{LD}}}
\newcommand{\finset}{{\operatorname{Fin}}}
\newcommand{\infset}{{\operatorname{Inf}}}

\newcommand{\finless}{\mathcal{F}}
\newcommand{\removeFin}{\operatorname{removeFin}}
\newcommand{\removeFinGBA}{\operatorname{removeFin}_{\operatorname{GBA}}}
\newcommand{\lifted}[2]{\uparrow_{#1}\!\!(#2)}
\newcommand{\GBAsum}{\oplus_{GBA}}
\newcommand{\bigGBAsum}[1]{{\bigoplus_{#1}}_{\raisebox{10pt}{\scriptsize  \!\!\!\!GBA}}}
\newcommand{\spot}{\textsc{Spot}}
\newcommand{\owl}{\textsc{Owl}}
\newcommand{\seminator}{\textsc{Seminator}}
\newcommand{\length}[1]{\vert #1 \vert}
\newcommand{\trueSymbol}{t\!t}		%
\newcommand{\falseSymbol}{f\!\!f}		%
\newcommand{\spin}{\textsc{Spin}}
\newcommand{\determin}{\operatorname{det}}
\newcommand{\splitTELA}{\operatorname{split}}
\newcommand{\splitTELAi}[2]{\operatorname{split}(#1)[#2]}
\newcommand{\mdplimit}{\operatorname{Limit}}

\newcommand{\approachOne}{remFin$\rightarrow$split$\alpha$}	%
\newcommand{\approachTwo}{split$\alpha$$\rightarrow$remFin}
\newcommand{\approachThree}{remFin$\rightarrow$rewrite$\alpha$}

\newcommand{\infsetColor}[2]{\infset (\color{#1} \bm{#2} \color{black})}

\renewcommand{\Pr}{\mathrm{Pr}}

\renewcommand{\S}{\mathfrak{S}}

\DeclareMathOperator{\IFF}{ iff }

\usepackage{tikz}
\usetikzlibrary{positioning,arrows,automata,calc,shapes}

\definecolor{darkgreen}{rgb}{0.0, 0.5, 0.0}
\definecolor{darkred}{rgb}{0.9, 0.0, 0.0}
\definecolor{darkblue}{rgb}{0.0, 0.0, 0.9}

\tikzset{
	state/.style={rounded rectangle,draw=black,inner sep=1.5mm,minimum width=9mm,minimum height=5mm},
	rstate/.style={rectangle,draw=black,inner sep=1ex,minimum width=9mm,minimum height=5mm},
	istate/.style={minimum width=5mm, minimum height=6mm},
	bullet/.style={circle,draw=black,fill=black,inner sep=0cm, minimum size=0.7mm},
	initial text={},
	every initial by arrow/.style={->,>=stealth'},
	ptran/.style={rounded corners, ->,>=stealth',auto},
	ntran/.style={rounded corners, -,auto}
}
 
\begin{document}

\newcommand{\orcid}[1]{\href{https://orcid.org/#1}{\textcolor[HTML]{A6CE39}{\aiOrcid}}}

\title{Determinization and Limit-determinization of Emerson-Lei automata}
\author{Tobias John\orcidID{0000-0001-5855-6632}\and Simon Jantsch\orcidID{0000-0003-1692-2408}\and \\ Christel Baier\orcidID{0000-0002-5321-9343}\and Sascha Klüppelholz\orcidID{0000-0003-1724-2586}}
\authorrunning{Tobias John, Simon Jantsch, Christel Baier and Sascha Kl\"uppelholz}
\titlerunning{}
\institute{Technische Universität Dresden\thanks{This work was funded by DFG grant 389792660 as part of \href{https://perspicuous-computing.science}{TRR~248}, the Cluster of Excellence EXC 2050/1 (CeTI, project ID 390696704, as part of Germany’s Excellence Strategy), DFG-projects BA-1679/11-1 and BA-1679/12-1, and the Research Training Group QuantLA (GRK 1763).}\\
\email{tobiasj@posteo.de \\ \{simon.jantsch, christel.baier, sascha.klueppelholz\}@tu-dresden.de}}

\maketitle

\begin{abstract}
  We study the problem of determinizing $\omega$-automata whose acceptance condition is defined on the transitions using Boolean formulas, also known as \emph{transition-based Emerson-Lei automata} (TELA).
  The standard approach to determinize TELA first constructs an equivalent \emph{generalized B\"uchi automaton} (GBA), which is later determinized.
  We introduce three new ways of translating TELA to GBA.
  Furthermore, we give a new determinization construction which determinizes several GBA separately and combines them using a product construction.
  An experimental evaluation shows that the product approach is competitive when compared with state-of-the-art determinization procedures.
  We also study limit-determinization of TELA and show that this can be done with a single-exponential blow-up, in contrast to the known double-exponential lower-bound for determinization.
  Finally, one version of the limit-determinization procedure yields \emph{good-for-MDP} automata which can be used for quantitative probabilistic model checking.
\end{abstract}

\section{Introduction}
Automata on infinite words, also called $\omega$-automata, play a fundamental role in the fields of verification and synthesis of reactive systems~\cite{PnueliR1989,Vardi1985a,CourcoubetisY1995,SafraV1989}.
They can be used both to represent properties of systems and the systems themselves.
For logical specification languages such as linear temporal logic (LTL), many verification systems, such as \spin{}~\cite{Ben-Ari2008} or \prism{}~\cite{KwiatkowskaNP2011}, use logic-to-automata translations internally to verify a given system against the specification.

A major research question in this area has been, and still is, the question of whether and how $\omega$-automata can be determinized efficiently~\cite{LodingP2019a,ScheweV2012,Safra1989,Redziejowski2012,MullerS1995}.
The first single-exponential and asymptotically optimal determinization for B\"uchi automata was presented in\cite{Safra1989}.
Deterministic automata are important from a practical point of view as classical automata-based solutions to reactive synthesis and probabilistic verification use deterministic automata~\cite{PnueliR1989,Vardi1985a}.

The high complexity of determinization and most logic-to-automata translations have raised the question of more succinct representations of $\omega$-automata.
Using \emph{generalized} acceptance conditions (e.g. generalized B\"uchi~\cite{Couvreur1999} or generalized Rabin\cite{BloemenDv2019,ChatterjeeGK2013}) and transition-based~\cite{GiannakopoulouL2002}, rather than state-based, conditions are common techniques in this direction.
An even more general approach has led to the HOA-format~\cite{BabiakBDKKMPS2015}, which represents the acceptance condition as a positive Boolean formula over standard B\"uchi ($\infset$) and co-B\"uchi ($\finset$) conditions, also called Emerson-Lei conditions~\cite{AllenEmersonL1987,SafraV1989}. 
Together with a vast body of work on heuristics and dedicated procedures this standardized format has led to practically usable and mature tools and libraries such as \spot{}~\cite{Duret-LutzLFMRX2016} and \owl{}~\cite{KretinskyMS2018} which support a wide range of operations on $\omega$-automata.
Special classes of nondeterministic automata with some of the desired properties of deterministic automata have also been studied.
The classes of good-for-MDP\cite{HahnPSSTW2020} and good-for-games\cite{KleinMBK2014} automata can be used for quantitative probabilistic model checking of Markov decision processes\cite{SickertEJK2016,HahnLSTZ2015}, while limit-deterministic B\"uchi automata can be used for qualitative model-checking~\cite{CourcoubetisY1995}.
Dedicated translations from LTL directly to deterministic and limit-deterministic automata have been considered in~\cite{EsparzaKS2018}.

This paper considers determinization and limit-determinization of TELA.
In contrast to limit-determinization, the theoretical complexity of determinization is well understood (a tight, double-exponential, bound was given in\cite{SafraV1989,Safra1989}).
However, it has not been studied yet from a practical point of view.

\noindent\textbf{Contribution.}
We propose three new translations from TELA to GBA (\Cref{sec:viagba}) and give an example in which they perform exponentially better than state-of-the-art implementations.
We introduce a new determinization procedure for TELA based on a product construction (\Cref{sec:determinize}). Our experiments show that it often outperforms the approaches based on determinizing a single GBA (\Cref{sec:evaluation}).
A simple adaptation of the product construction produces \emph{limit-deterministic} TELA of single-exponential size (in contrast to the double-exponential worst-case complexity of full determinization, \Cref{sec:limitdetviaprod}).
We show that deciding $\prb^{\max}_{\M}(\lang(\mathcal{A})) > 0$ is NP-complete for limit-deterministic TELA $\mathcal{A}$, and in P if the acceptance of $\mathcal{A}$ is fin-less (\Cref{prop:limitdetqual}).
Finally, we show how to limit-determinize TELA based on the breakpoint-construction.
A version of this procedure yields \emph{good-for-MDP} B\"uchi automata (\Cref{def:gfmdef}).
Thereby $\prb^{\max}_{\M}(\lang(\mathcal{A}))$ is computable in single-exponential time for arbitrary MDP $\M$ and TELA $\mathcal{A}$ (\Cref{cor:quantsingleexp}).

\noindent\textbf{Related work.}
The upper-bound for TELA-determinization \cite{SafraV1989,Safra1989} relies on a translation to GBA which first transforms the acceptance formula into disjunctive normal form (DNF). We build on this idea.
Another way of translating TELA to GBA was described in~\cite{DuretLutz2017}.
Translations from LTL to TELA have been proposed in~\cite{MuellerS2017,MajorBFSSZ2019,BlahoudekMS2019a}, and all of them use product constructions to combine automata for subformulas.
The emptiness-check for $\omega$-automata under different types of acceptance conditions has been studied in~\cite{BaierBDKMS2019,DuretLutzJMZ2009,BloemenDv2019,ChatterjeeGK2013}, where~\cite{BaierBDKMS2019} covers the general case of Emerson-Lei conditions and also considers qualitative probabilistic model checking using deterministic TELA.
The generalized Rabin condition from~\cite{BloemenDv2019,ChatterjeeGK2013} is equivalent to the special DNF that we use and a special case of the hyper-Rabin condition for which the emptiness problem is in P~\cite{AllenEmersonL1987,Boker2018}.
Probabilistic model checking for \emph{deterministic} automata under this condition is considered in\cite{ChatterjeeGK2013}, while\cite{BloemenDv2019} is concerned with standard emptiness while allowing nondeterminism.
A procedure to transform TELA into parity automata is presented in~\cite{RenkinDP2020}.

\section{Preliminaries}

\textbf{Automata.}
A \emph{transition-based Emerson-Lei} automaton (TELA) $\mathcal{A}$ is a tuple $(Q, \Sigma, \delta, I, \alpha)$, where $Q$ is a finite set of states, $\Sigma$ is a finite alphabet, $\delta \subseteq Q \times \Sigma \times Q$ is the transition relation, $I \subseteq Q$ is the set of initial states and $\alpha$ is a symbolic acceptance condition over $\delta$, which is defined by:
\vskip-2ex
	\[ \alpha ::=  \trueSymbol \mid \falseSymbol \mid \infset(T) \mid \finset(T) \mid (\alpha \lor \alpha) \mid (\alpha \land \alpha),  \text{ with } T\subseteq \delta  \]
\vskip-0.5ex
\noindent
If $\alpha$ is $\trueSymbol$, $\falseSymbol$, $\infset(T)$ or $\finset(T)$, then it is called \emph{atomic}. We denote by $\length{\alpha}$ the number of atomic conditions contained in $\alpha$, where multiple occurrences of the same atomic condition are counted multiple times. 
Symbolic acceptance conditions describe sets of transitions $T \subseteq \delta$. Their semantics is defined recursively as follows:
\vskip-4ex
\begin{alignat*}{6}
	&{T} && \models \trueSymbol  &~~~~~{T} & \models \infset(T') &~\IFF~ {T} \cap T' \neq \emptyset && ~~~~~{T} & \models \alpha_1 \lor \alpha_2 &&~\IFF~ {T} \models \alpha_1 \text{ or } {T} \models \alpha_2\\
	&{T} && \not\models \falseSymbol  & {T} & \models \finset(T') &~\IFF~ {T} \cap T' = \emptyset && {T} & \models \alpha_1 \land \alpha_2  &&~\IFF~ {T} \models \alpha_1 \text{ and } {T} \models \alpha_2
\end{alignat*}
\vskip-1ex
\noindent
Two acceptance conditions $\alpha$ and $\beta$ are $\delta$-\emph{equivalent} ($\alpha \equiv_{\delta} \beta$) if for all $T \subseteq \delta$ we have $T \models \alpha \iff T \models \beta$.
A \emph{run} of $\mathcal{A}$ for an infinite word $u = u_0 u_1 \ldots \in \Sigma^\omega$ is an infinite sequence of transitions $\rho = (q_0, u_0, q_1)(q_1, u_1, q_2)\ldots \in \delta^\omega$ that starts with an initial state $q_0 \in I$. The set of transitions that appear infinitely often in $\rho$ are denoted by $\inf(\rho)$. A run $\rho$ is \emph{accepting} ($\rho \models \alpha$) $\IFF$ $\inf(\rho) \models \alpha$. The language of $\mathcal{A}$, denoted by $\lang(\mathcal{A})$, is the set of all words for which there exists an accepting run of $\mathcal{A}$.
The sets of infinite words which are the language of some TELA are called \emph{$\omega$-regular}.
A TELA $\mathcal{A}$ is \emph{deterministic} if the set of initial states contains exactly one state and the transition relation is a function $\delta: Q \times \Sigma \to Q$.
It is \emph{complete}, if for all $(q,a) \in Q \times \Sigma$: $\delta \cap \{(q,a,q') \mid q' \in Q\} \neq \varnothing$.
A \emph{B\"uchi condition} is an acceptance condition of the form $\infset(T)$ and a \emph{generalized B\"uchi condition} is a condition of the form $\bigwedge_{1 \leq i \leq k} \infset(T_i)$.
We call the sets $T_i$ appearing in a generalized B\"uchi condition its \emph{acceptance sets}.
Rabin (resp. Street) conditions are of the form $\bigvee_{1 \leq i \leq k} (\finset(F_i) \land \infset(T_i))$ (resp. $\bigwedge_{1 \leq i \leq k} (\finset(F_i) \lor \infset(T_i))$).

\vskip+1ex
\noindent\textbf{Probabilistic systems.}
A labeled \emph{Markov decision process} (MDP) $\M$ is a tuple $(S,s_0,\Act,P,\Sigma,L)$ where $S$ is a finite set of states, $s_0 \in S$ is the initial state, $\Act$ is a finite set of actions, $P : S \times \Act \times S \to [0,1]$ defines the transition probabilities with $\sum_{s' \in S} P(s,\alpha,s') \in \{0,1\}$ for all $(s,\alpha) \in S \times \Act$ and $L : S \to \Sigma$ is a labeling function of the states into a finite alphabet $\Sigma$.
Action $\alpha \in \Act$ is \emph{enabled} in $s$ if $\sum_{s' \in S} P(s,\alpha,s') = 1$, and $\Act(s) = \{\alpha \mid \alpha \text{ is enabled in } s\}$.
A \emph{path} of $\M$ is an infinite sequence $s_0 \alpha_0 s_1 \alpha_1 \ldots \in (S \times \Act)^{\omega}$ such that $P(s_i,\alpha_i,s_{i+1}) > 0$ for all $i \geq 0$.
The set of all paths of $\M$ is denoted by $\Paths(\M)$ and $\Pathsfin(\M)$ denotes the \emph{finite paths}.
Given a path $\pi = s_0 \alpha_0 s_1 \alpha_1 \ldots$, we let $L(\pi) = L(s_0)L(s_1) \ldots \in \Sigma^{\omega}$.
A \emph{Markov chain} is an MDP with $|\Act(s)| \leq 1$ for all states $s$.
A scheduler of $\M$ is a function $\mathfrak{S}: (S\times\Act)^*\times S \rightarrow \Act$ such that $\mathfrak{S}(s_0 \alpha_0 \ldots s_n) \in \Act(s_n)$.
It induces a Markov chain $\M_{\mathfrak{S}}$ and thereby a probability measure over $\Paths(\M)$.
The probability of a set of paths $\Pi$ starting in $s_0$ under this measure is $\Pr_\M^{\mathfrak{S}}(\Pi)$.
For an $\omega$-regular property $\Phi \subseteq \Sigma^{\omega}$ we define $\prb^{\max}_{\M}(\Phi) = \sup_{\mathfrak{S}}\Pr_{\M}^{\mathfrak{S}}(\{\pi \mid \pi \in \Paths(\M) \text{ and } L(\pi) \in \Phi\})$.
See\cite[Chapter 10]{BaierK2008} for more details.

\section{From TELA to Generalized Büchi Automata}
\label{sec:viagba}

\subsection{Operations on Emerson-Lei Automata}
The first operator takes a TELA and splits it along the top-level disjuncts of the acceptance condition.
Let $\mathcal{A} = (Q, \Sigma, \delta, I, \alpha)$ be a TELA where $\alpha = \bigvee_{1 \leq i \leq m} \alpha_i$ and the $\alpha_i$ are arbitrary acceptance conditions.
We define $\splitTELA(\mathcal{A}) := (\mathcal{A}_1, \ldots, \mathcal{A}_m)$ with $\mathcal{A}_i = (Q, \Sigma, \delta, I, \alpha_i)$ for $1 \leq i \leq m$, and $\splitTELAi{\mathcal{A}}{i} := \mathcal{A}_i$.
\begin{restatable}{lemma}{splitConj}
\label{lem:splitAcceptance}
It holds that
	$\lang(\mathcal{A}) = \bigcup_{1 \leq i \leq m} \lang\big(\splitTELAi{\mathcal{A}}{i}\big)$.
\end{restatable}
\noindent
The analogous statement does not hold for conjunction and intersection (cf~\Cref{fig:splitconj} in the appendix).
We also need constructions to realize the union of a sequence of automata.
This can either be done using the sum (i.e. disjoint union) or the disjunctive product of the state spaces.
We define a general sum (simply called \emph{sum}) operation and one that preserves GBA acceptance (called \emph{GBA-specific sum}).
The disjunctive product construction for TELA is mentioned in \cite{DuretLutz2017} and similar constructions are used in \cite{MuellerS2017, MajorBFSSZ2019}.
While the sum operations yield smaller automata in general, only the product construction preserves determinism.
\begin{definition}\label{def:combineTELA}
	Let $\mathcal{A}_i = (Q_i, \Sigma, \delta_i, I_i, \alpha_i)$, with $i \in \{0,1\}$, be two complete TELA with disjoint state-spaces.
	The \emph{sum} of $\mathcal{A}_0$ and $\mathcal{A}_1$ is defined as follows: 
    \vskip-4ex
		\[\mathcal{A}_0 \oplus \mathcal{A}_1 = \big(Q_0 \cup Q_1, \Sigma, \delta_0 \cup \delta_1, I_0 \cup I_1, (\alpha_0 \land \infset(\delta_0)) \lor (\alpha_1 \land \infset(\delta_1) \big)\]
	\vskip-1ex
	\noindent If $\alpha_i = \infset(T_1^i) \wedge \ldots \wedge \infset(T_k^i)$, with $i \in \{0,1\}$, (i.e. both automata are GBA), then we can use the \emph{GBA-specific sum}:
    \vskip-4ex
		\[\mathcal{A}_0 \GBAsum \mathcal{A}_1 = \big(Q_0 \cup Q_1, \Sigma, \delta_0 \cup \delta_1, I_0 \cup I_1, (\infset(T_1^0 \cup T_1^1) \wedge \ldots \wedge \infset(T_k^0 \cup T_k^1)) \big)\]
	\vskip-1ex
	\noindent
		The \emph{disjunctive  product}  is defined as follows:
	    \vskip-2ex
		\[\mathcal{A}_0 \otimes \mathcal{A}_1 = \big(Q_0 \times Q_1, \Sigma, \delta_\otimes, I_0 \times I_1, (\lifted{}{\alpha_0} \vee \lifted{}{\alpha_1})\big)\]
	\vskip-0.5ex
	\noindent		
	with
	$\delta_\otimes = \big\lbrace \big((q_0, q_1), a, (q_0', q_1')\big) \bigm\vert (q_0,a,q_0') \in \delta_0 \text{ and } (q_1,a,q_1') \in \delta_1 \big\rbrace$ and 
	${\lifted{}{\alpha_i}}$ is constructed by replacing every occurring set of transitions $T$ in $\alpha_i$ by $\big\lbrace \big((q_0, q_1), u, (q_0', q_1') \big) \in \delta_\otimes \bigm\vert (q_i, u, q_i') \in T \big\rbrace$.
\end{definition}

The additional $\infset(\delta_0)$ and $\infset(\delta_1)$ atoms in the acceptance condition of $\mathcal{A}_0 \oplus \mathcal{A}_1$ are essential (see~\Cref{fig:SumWithoutInf} in the appendix). 
We can apply the GBA-specific sum to any two GBA by adding $\infset(\delta_i)$ atoms until the acceptance conditions are of equal length.
Many of our constructions will require the acceptance condition of the TELA to be in DNF.
We will use the following normal form throughout the paper (also called \emph{generalized Rabin} in\cite{ChatterjeeGK2013,BloemenDv2019}).

\begin{definition}[DNF for TELA]
  \label{def:dnf}
  Let $\mathcal{A} = (Q, \Sigma,\delta, I,\alpha)$ be a TELA. We say that $\mathcal{A}$ is in \emph{DNF} if $\alpha$ is of the form $\alpha = \bigvee_{1 \leq i \leq m} \alpha_i$, with $\alpha_i = \finset(T^i_0) \land \bigwedge_{1 \leq j \leq k_i} \infset(T^i_j)$ and such that all $k_i \geq 1$.
\end{definition}

The reason that a single $\finset$ atom in each disjunct is enough is that $\finset(T_1) \land \finset(T_2) \equiv_{\delta} \finset(T_1 \cup T_2)$ for all $T_1,T_2,\delta$.
Taking $k_i \geq 1$ is also no restriction, as we can always add $\land \, \infset(\delta)$ to any disjunct.
Using standard Boolean operations one can transform a TELA with acceptance $\beta$ into DNF by just translating the acceptance formula into a formula $\alpha$ of the above form, with $|\alpha| \leq 2^{|\beta|}$.

\vskip+2ex
\noindent\textbf{Fin-Less Acceptance.}
\begin{figure}[tbp]
    \centering
	\includegraphics[width=0.86\textwidth]{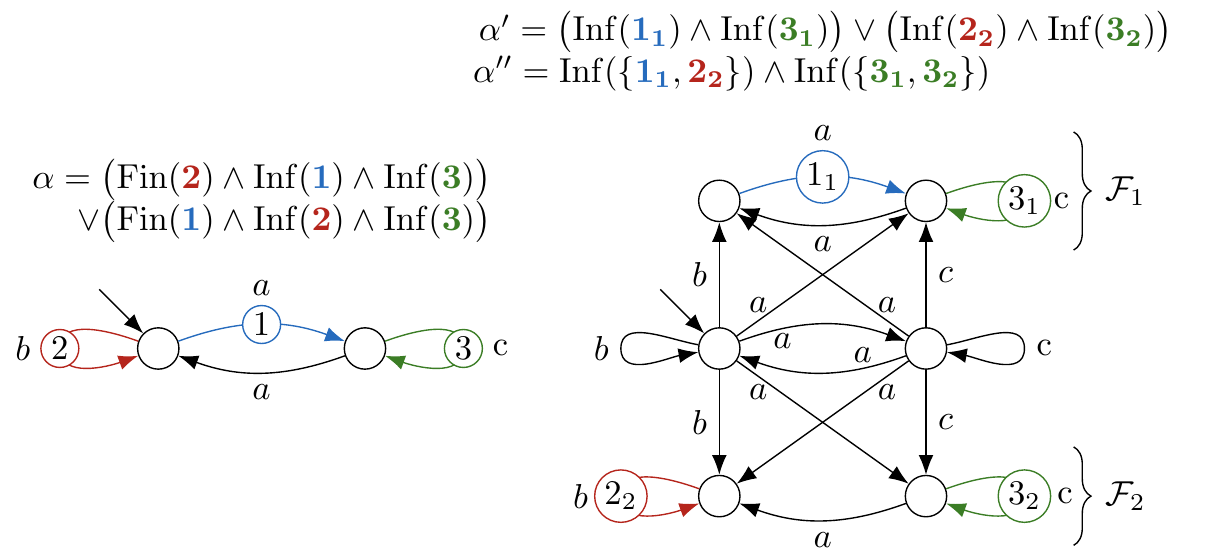}
	\caption{Example of applying $\removeFin$ and $\removeFinGBA$ (\Cref{def:removeFin}) to the automaton on the left. The result is the automaton on the right with acceptance $\alpha'$ ($\removeFin$), respectively $\alpha''$ ($\removeFinGBA$).}
	\label{fig:removeFin}
\end{figure}
To transform a TELA in DNF (see~\Cref{def:dnf}) into an equivalent one without $\finset$-atoms we use the idea of\cite{DuretLutz2017,BloemenDv2019}: a main copy of $\mathcal{A}$ is connected to one additional copy for each disjunct $\alpha_i$ of the acceptance condition, in which transitions from $T_0^i$ are removed.
The acceptance condition ensures that every accepting run leaves the main copy eventually.
Figure~\ref{fig:removeFin} shows an example.

\begin{definition}\label{def:removeFin}
	Let $\finless_i = (Q_i, \Sigma, \delta_i, I_i, \phi_i)$, where $Q_i = \{q^{(i)} \mid q \in Q\}$, $\delta_i = \{(q^{(i)},a,{q'^{(i)}}) \mid (q,a,q') \in \delta \setminus T^i_0\}$ and $\phi_i = \bigwedge_{1 \leq j \leq k_i} \infset(U^i_j)$, where $U^i_j = \{(q^{(i)},a,{q'^{(i)}}) \mid (q,a,q') \in T_j^i \setminus T^i_0\}$.
	Let $\removeFin(\mathcal{A}) = (Q', \Sigma, \delta', I, \alpha')$ and $\removeFinGBA(\mathcal{A}) = (Q', \Sigma, \delta', I, \alpha'')$, where $Q' = Q \cup \bigcup_{1 \leq i \leq m} Q_{i}$ and:
    \vskip-2ex
	\begin{itemize}
		\item[$\bullet$] $\delta' = \delta \cup \bigcup_{1 \leq i \leq m}\big(\delta_{i} \cup \{(q, a, {q'^{(i)}}) \mid (q, a, q') \in \delta \} \big)$
		\item[$\bullet$] $\alpha' ~\!= \bigvee_{1 \leq i \leq m} \phi_i$ %
		\item[$\bullet$] $\alpha'' = \bigwedge_{1\leq j \leq k} \infset(U_j^1 \cup \ldots \cup U_j^m)$, with $k = \max_{i} k_i$ and $U_j^i = \delta_i$ if $k_i < j \leq k$.
	\end{itemize}
\end{definition}

\begin{restatable}{lemma}{removeFinGba}\label{lem:removeFinGBA}	
	It holds that $\lang(\mathcal{A}) = \lang(\removeFin(\mathcal{A})) = \lang(\removeFinGBA(\mathcal{A}))$.
\end{restatable}
\noindent
While $\removeFin(\mathcal{A})$ is from\cite{DuretLutz2017,BloemenDv2019}, $\removeFinGBA(\mathcal{A})$ is a variant that differs only in the acceptance and always produces GBA.
Both consist of $m+1$ copies of $\mathcal{A}$ (with $\finset$-transitions removed).

\subsection{Construction of Generalized Büchi Automata}
\begin{figure}[btp]
	\centering
	\[\alpha_n =\big( \infsetColor{NavyBlue}{1} \wedge \infsetColor{CornflowerBlue}{1'} \big) \vee \big( \infsetColor{BrickRed}{2} \wedge \infsetColor{Red}{2'} \big) \vee \ldots \vee \big(\infsetColor{OliveGreen}{n} \vee \infsetColor{LimeGreen}{n'} \big) \vspace{-10pt}\] 
	\includegraphics[width=0.86\textwidth]{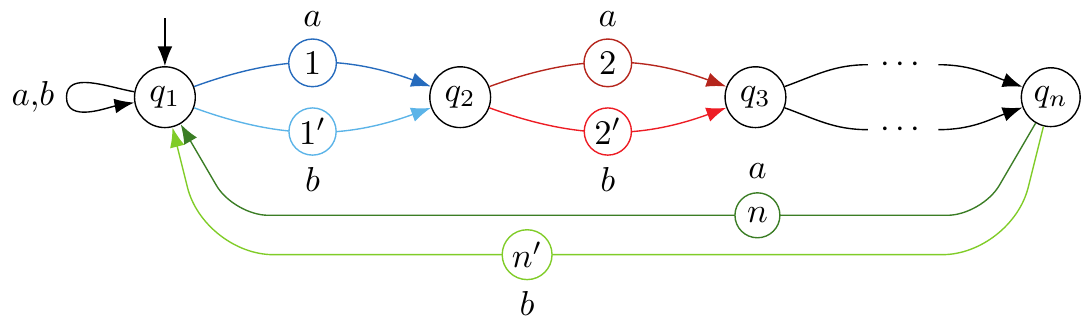}
	\caption{A class of {TELA} where generating the {CNF} leads to $2^n$ many conjuncts.}
	\label{fig:ExpCNF}
\end{figure}

\noindent\textbf{Construction of Spot.}
The transformation from TELA to GBA from\cite{DuretLutz2017} is implemented in \spot \cite{Duret-LutzLFMRX2016}.
It transforms the automaton into DNF and then applies (an optimized version of) $\removeFin$. 
The resulting fin-less acceptance condition is translated into conjunctive normal form ({CNF}).
As $\infset(T_1) \lor \infset(T_2) \equiv_{\delta} \infset(T_1 \cup T_2)$ holds for all $\delta$, one can rewrite any fin-less condition in CNF into a conjunction of $\infset$-atoms, which is a generalized B\"uchi condition.
When starting with a TELA $\cal B$ with acceptance $\beta$ and $N$ states, one gets a GBA with $N \, 2^{O(|\beta|)}$ states and $2^{O(|\beta|)}$ acceptance sets, as the fin-removal may introduce exponentially (in $|\beta|$) many copies, and the CNF may also be exponential in $|\beta|$.

Transforming a fin-less automaton into a GBA by computing the CNF has the advantage of only changing the acceptance condition, and in some cases it produces simple conditions directly. 
For example, \spot's TELA to GBA construction transforms a Rabin into a B\"uchi automaton, and a Streett automaton with $m$ acceptance pairs into a GBA with $m$ accepting sets.
However, computing the CNF may also incur an exponential blow-up (\Cref{fig:ExpCNF} shows such an example).

\vskip+2ex
\noindent\textbf{Copy-based approaches.}
We now describe three approaches (\approachOne{}, \approachTwo{} and \approachThree{}), which construct GBA with at most $|\beta|$ acceptance sets.
On the other hand, they generally produce automata with more states.
They are based on\cite{SafraV1989} which first translates copies of $\mathcal{A}$ (corresponding to the disjuncts of the acceptance condition) to GBA, and then takes their sum.
However, it is not specified in\cite{SafraV1989} how exactly $\finset$-atoms should be removed (they were concerned only with the theoretical complexity). We define:
\begin{align*}
  \text{\approachOne{}}(\mathcal{A}) &:= \bigGBAsum{1 \leq i \leq m} \splitTELAi{\removeFin(\mathcal{A})}{i} \\
  \text{\approachTwo{}}(\mathcal{A}) &:= \bigGBAsum{1 \leq i \leq m} \removeFin(\splitTELAi{\mathcal{A}}{i}) \\
  \text{\approachThree{}}(\mathcal{A}) &:= \removeFinGBA(\mathcal{A})
\end{align*}
With $\removeFin$ as defined in~\Cref{def:removeFin}, the approaches \approachOne{} and \approachTwo{} produce the same automata (after removing non-accepting SCC's in \approachOne{}), and all three approaches create $O(m)$ copies of $\mathcal{A}$.
Our implementation uses an optimized variant of $\removeFin$,
as provided by \spot{}, which leads to different results for all three approaches.

\section{Determinization}
\label{sec:determinize}

\noindent\textbf{Determinization via single GBA.}
The standard way of determinizing TELA is to first construct a GBA, which is then determinized.
Dedicated determinization procedures for GBA with $N$ states and $K$ acceptance sets produce deterministic Rabin automata with $2^{O(N (\log N + \log K))}$ states\cite{ScheweV2012}.
For a TELA $\mathcal{B}$ with $n$ states and acceptance $\beta$, the above translations yield GBA with $N = n \, 2^{O(|\beta|)}$ and $K = 2^{O(|\beta|)}$ (\spot{}'s construction) or $N = n \, 2^{O(|\beta|)}$ and $K = O(|\beta|)$ (copy-based approaches).
We evaluate the effect of the translations to GBA introduced in the previous chapter in the context of determinization in~\Cref{sec:evaluation} .

\vskip+2ex
\noindent\textbf{Determinization via a product construction.}
Another way to determinize a TELA $\mathcal{A}$ in DNF is to determinize the automata $\splitTELAi{\mathcal{A}}{i}$ one by one and then combining them with the disjunctive product construction of \Cref{def:combineTELA}:
\vskip-1ex
\[ \bigotimes_{1 \leq i \leq m} \determin\big(\removeFin(\splitTELAi{\mathcal{A}}{i})\big) \]
\vskip-0.5ex
\noindent
where ``$\determin$'' is any GBA-determinization procedure.
Let $\cal B$ be a TELA with acceptance $\beta$ and $n$ states, and let $\alpha$ be an equivalent condition in DNF with $m$ disjuncts.
Assuming an optimal GBA-determinization procedure, the product combines $m$ automata with $2^{O(n (\log n + \log |\beta|))}$ states and hence has $\big(2^{O(n \, (\log n + \log |\beta|))}\big)^m = 2^{O ( 2^{|\beta|} \cdot n (\log n + \log |\beta|))}$ states.

\section{Limit-deterministic TELA}
\label{sec:limitdet}
Limit-determinism has been studied mainly in the context of B\"uchi automata~\cite{CourcoubetisY1995,SickertEJK2016,Vardi1985a}, and we define it here for general TELA.
\begin{definition}
  \label{def:limitdet}
  A TELA $\mathcal{A} = (Q,\Sigma,\delta,I,\alpha)$ is called \emph{limit-deterministic} if there exists a partition $Q_N,Q_D$ of $Q$ such that
  \begin{enumerate}
  \item $\delta \cap (Q_D \times \Sigma \times Q_N) = \varnothing$,
  \item for all $(q,a) \in Q_D \times \Sigma$ there exists at most one $q'$ such that $(q,a,q') \in \delta$,
  \item every accepting run $\rho$ of $\mathcal{A}$ satisfies $\inf(\rho) \cap (Q_N \times \Sigma \times Q_N) = \varnothing$.
  \end{enumerate}
\end{definition}
\noindent
This is a semantic definition and as checking emptiness of deterministic TELA is already coNP-hard, checking whether a TELA is limit-determinstic is also.
\begin{restatable}{proposition}{LimitDetComplexity}
  \label{prop:limitdetcoml}
  Checking limit-determinism for TELA is coNP-complete.
\end{restatable}
\noindent An alternative syntactic definition for TELA in DNF, which implies limit-determinism, 
is provided in \Cref{def:syntlimitdet}.
\begin{definition}
  \label{def:syntlimitdet}
  A TELA $\mathcal{A} = (Q,\Sigma,\delta,\{q_0\},\alpha)$ in DNF, with $\alpha = \bigvee_{1 \leq i \leq m} \alpha_i$, $\alpha_i = \finset(T^i_0) \land \bigwedge_{1 \leq j \leq k_i} \infset(T^i_j)$ and $k_i \geq 1$ for all $i$, is \emph{syntactically limit-deterministic} if there exists a partition $Q_N,Q_D$ of $Q$ satisfying 1. and 2. of~\Cref{def:limitdet} and additionally $T^i_j \subseteq Q_D \times \Sigma \times Q_D \text{ for all } i \leq m \text{ and } 1 \leq j \leq k_i$.
\end{definition}

\subsection{Limit-determinization}
\label{sec:limitdetviaprod}
We first observe that replacing the product by a sum in the product-based determinization above yields limit-deterministic automata of single-exponential size (in contrast to the double-exponential lower-bound for determinization).
Let $\mathcal{A}$ be a TELA in DNF with $n$ states and acceptance $\alpha = \bigvee_{1 \leq i \leq m} \alpha_i$, where $\alpha_i = \finset(T^i_0) \land \bigwedge_{1 \leq j \leq k_i} \infset(T^i_j)$ (see \Cref{def:dnf}), and let $\mathcal{A}_i = \splitTELAi{\mathcal{A}}{i}$.
\begin{restatable}{proposition}{LimitDetApprOne}
\label{prop:limitdet1}
  $\bigoplus_{1 \leq i \leq m} \determin(\removeFin(\mathcal{A}_i))$ is limit-deterministic and of size $\sum_{1 \leq i \leq m} |\determin(\mathcal{A}_i)| = m \cdot 2^{O(n \, (\log n + \log k))}$, where $k = \max \{k_i \mid 1 \leq i \leq m\}$.
\end{restatable}
\noindent If ``$\det$'' is instantiated by a GBA-determinization that produces Rabin automata, then the result is in DNF and syntactically limit-determinstic.
Indeed, in this case the only nondeterminism is the choice of the initial state.
But ``$\det$'' can, in principle, also  be replaced by any limit-determinization procedure for GBA.

We now extend the limit-determinization constructions of\cite{CourcoubetisY1995} (for B\"uchi automata) and\cite{HahnLSTZ2015,BlahoudekDKKS2017,Blahoudek2018} (for GBA) to Emerson-Lei conditions in DNF.
These constructions use an \emph{initial} component and an \emph{accepting breakpoint component}\cite{MiyanoH1984} for $\mathcal{A}$, which is deterministic.
The following construction differs in two ways: there is one accepting component per disjunct of the acceptance condition, and the accepting components are constructed from $\mathcal{A}$ without considering the $\finset$-transitions of that disjunct.
To define the accepting components we use the subset transition function $\theta$ associated with $\delta$:
$\theta(P,a) = \bigcup_{q \in P}\{ q' \mid (q,a,q') \in \delta\}$ for $(P,a) \in 2^Q \times \Sigma$, and additionally we define $\theta|_T(P,a) = \bigcup_{q \in P}\{ q' \mid (q,a,q') \in \delta \cap T \}$.
These functions are extended to finite words in the standard way.

\begin{definition}
  \label{def:breakp}
	Let $\theta_i = \theta|_{\delta \setminus T_0^i}$ and define $\breakp_{i} = (Q_i,\Sigma,\delta_i,\{p_0\},\infset(\delta^{\operatorname{break}}_i))$ with: $Q_i = \{(R,B,l) \in 2^Q \times 2^Q \times \{0, \ldots, k_i\} \mid B \subseteq R \}$, $p_0 = (I,\varnothing,0)$ and
  \begin{align*}
    \delta^{\operatorname{main}}_i &= \left\lbrace((R_1, B_1, l), a, (R_2, B_2, l)) \mid\begin{array}{l}
		 R_2 = \theta_i (R_1, a), \\
		 B_2 = \theta_i (B_1, a) \cup \theta_i|_{T_l^i}(R_1, a)
		\end{array} \right\rbrace \\
    \delta^{\operatorname{break}}_i &= \left\lbrace ((R_1, B_1, l), a, (R_2, \varnothing, l') \mid \begin{array}{l}
		  ((R_1, B_1, l), a, (R_2, B_2, l)) \in \delta^{\operatorname{main}}_i,\\
      R_2 = B_2,\\
		  l' = (l+1) \bmod (k_i+1)
		\end{array} \right\rbrace \\
    \delta_i &= \bigl\{((R_1, B_1, l), a, (R_2, B_2, l)) \in \delta^{\operatorname{main}}_i \mid R_2 \neq B_2 \bigr\} \cup \delta^{\operatorname{break}}_i
  \end{align*}
\end{definition}

In state $(R,B,l)$, intuitively $R$ is the set of states reachable for the prefix word in $\mathcal{A}$ without using transitions from $T_0^i$, while $B$ are the states in $R$ which have seen a transition in $T_l^i$ since the last ``breakpoint''.
The breakpoint-transitions are $\delta^{\operatorname{break}}_i$, which occur when all states in $R$ have seen an accepting transition since the last breakpoint (namely if $R = B$).
The breakpoint construction underapproximates the language of a given GBA, in general.

We define two limit-deterministic B\"uchi automata (LDBA) $\ld_{\mathcal{A}}$ and $\gfm_{\mathcal{A}}$ where $\gfm_{\mathcal{A}}$ is additonally \emph{good-for-MDP} (GFM)\cite{HahnPSSTW2020}.
This means that $\gfm_{\mathcal{A}}$ can be used to solve certain quantitative probabilistic model checking problems (see~\Cref{sec:probmc}).
Both use the above breakpoint automata as accepting components.
While $\ld_{\mathcal{A}}$ simply uses a copy of $\mathcal{A}$ as initial component, $\gfm_{\mathcal{A}}$ uses the deterministic subset-automaton of $\mathcal{A}$ (it resembles the \emph{cut-deterministic} automata of~\cite{Blahoudek2018}).
Furthermore,  to ensure the GFM property, there are more transitions between initial and accepting copies in $\gfm_{\mathcal{A}}$.
The construction of $\gfm_{\mathcal{A}}$ extends the approach for GBA in~\cite{HahnLSTZ2015} (also used for probabilistic model checking) to TELA.
We will distinguish elements from sets $Q_i$ for different $i$ from \Cref{def:breakp} by using subscripts (e.g. $(R,P,l)_i$) and assume that these sets are pairwise disjoint.%
\begin{definition}[$\ld_{\mathcal{A}}$ and $\gfm_{\mathcal{A}}$]
  \label{def:gfmdef}
  Let $Q_{\operatorname{acc}} = \bigcup_{1 \leq i \leq m} Q_i$, $\delta_{\operatorname{acc}} = \bigcup_{1 \leq i \leq m} \delta_i$ and $\alpha_{\operatorname{acc}} = \infset(\bigcup_{1 \leq i \leq m} \delta^{\operatorname{break}}_i)$. Define
  \[\ld_{\mathcal{A}} = (Q \cup Q_{\operatorname{acc}},\Sigma,\ldtrans,I,\alpha') \; \text{ and } \; \gfm_{\mathcal{A}} = (2^{Q} \cup Q_{\operatorname{acc}},\Sigma,\gfmtrans,\{I\},\alpha')\]
  with 
  
  \vskip-4ex
  {\ }
  \vskip-4ex
  
  \begin{align*}
  \ldtrans & = \delta \cup \ldbridge \cup \delta_{\operatorname{acc}} \quad \text{ and } \quad \gfmtrans = \theta \cup \gfmbridge \cup \delta_{\operatorname{acc}}\\
  \ldbridge & = \bigl\{\bigl(q,a,(\{q'\},\varnothing,0)_i\bigr) \mid (q,a,q') \in \delta \text{ and } 1 \leq i \leq m \bigr\}\\[0.25ex]
  \gfmbridge & = \bigl\{\bigl(P,a,(P',\varnothing,0)_i\bigr) \mid P' \subseteq \theta(P,a) \text{ and } 1 \leq i \leq m \bigr\}
  \end{align*}
\end{definition}
\noindent
As $\delta_i^{\operatorname{break}} \subseteq \delta_{\operatorname{acc}}$ for all $i$, both $\ld_{\mathcal{A}}$ and $\gfm_{\mathcal{A}}$ are syntactically limit-deterministic.
The proofs of correctness are similar to ones of the corresponding constructions for GBA~\cite[Thm. 7.6]{Blahoudek2018}.
We show later in~\Cref{thm:gfmprop} that $\gfm_{\mathcal{A}}$ is GFM.
\begin{restatable}{theorem}{GFMcorrectness}
\label{thm:GFMcorrectness}
  $\ld_{\mathcal{A}}$ and $\gfm_{\mathcal{A}}$ are syntactically limit-deterministic and satisfy $\lang(\ld_{\mathcal{A}}) = \lang(\gfm_{\mathcal{A}}) = \lang(\mathcal{A})$. Their number of states is in $O(n +  3^n \, m \, k)$ for $\ld_{\mathcal{A}}$ and $O(2^n +  3^n \, m \, k) = O(|\alpha|^2 \cdot 3^n)$ for $\gfm_{\mathcal{A}}$, where $k = \max \{k_i \mid 1 \leq i \leq m\}$.
\end{restatable}
\begin{corollary}
  Given TELA $\cal B$ (not necessarily in DNF) with acceptance condition $\beta$ and $N$ states, there exists an equivalent LDBA with $2^{O(|\beta| + N)}$ states.
\end{corollary}
\subsection{Probabilistic model checking}
\label{sec:probmc}

We now discuss how these constructions can be used for probabilistic model checking.
First, consider the \emph{qualitative} model checking problem to decide $\prb^{\max}_{\M}(\lang(\mathcal{A})) > 0$, under the assumption that $\mathcal{A}$ is a limit-deterministic TELA.
While NP-hardness follows from the fact that the problem is already hard for deterministic TELA~\cite[Thm. 5.13]{Mueller2019}, we now show that it is also in NP.
Furthermore, it is in P for automata with a fin-less acceptance condition.
This was already known for LDBA~\cite{CourcoubetisY1995}, and our proof uses similar arguments.
\begin{restatable}{proposition}{LDTELAqualitative}
  \label{prop:limitdetqual}
  Deciding $\prb_{\mathcal{M}}^{\max}(\lang(\mathcal{A})) > 0$, given an MDP $\M$ and a limit-deterministic TELA $\mathcal{A}$, is NP-complete. If $\mathcal{A}$ has a fin-less acceptance condition, then the problem is in P.
\end{restatable}
Now we show that $\gfm_\mathcal{A}$ is good-for-MDP~\cite{HahnPSSTW2020}.
In order to define this property, we introduce the product of an MDP with a nondeterministic automaton in which, intuitively, the scheduler is forced to resolve the nondeterminism by choosing the next state of the automaton (see~\cite{KleinMBK2014,HahnPSSTW2020}).
We assume that the automaton used to build the product has a single initial state, which holds for $\gfm_\mathcal{A}$.
\begin{definition}
  \label{def:product}
  Given an MDP $\M = (S,s_0,\Act,P,\Sigma,L)$ and TELA $\mathcal{G} = (Q,\Sigma,\delta,\{q_0\},\alpha)$ we define the MDP $\M \times \mathcal{G} = (S \times Q, (s_0,q_0), \Act \times Q, P^{\times}, \Sigma, L^{\times})$ with $L^{\times}((s,q)) = L(s)$ and
  \vskip-3.5ex
  \[ P^{\times}\bigl((s,q),(\alpha,p),(s',q')\bigr) = 
    \begin{cases}
      P(s,\alpha,s') & \text{if } p = q' \text{ and } (q,L(s),q') \in \delta \\
      0 & \text{otherwise}
    \end{cases}\]
\end{definition}
\vskip-0.5ex
\noindent
We define the \emph{accepting paths} $\Pi_{acc}$ of $\M \times \mathcal{G}$ to be:
 \vskip-4ex
 \[\Pi_{acc} = \{(s_0,q_0)\alpha_0(s_1,q_1)\alpha_1 \ldots \in \Paths(\M \times \mathcal{G}) \mid q_0, L(s_0), q_1, L(s_1) \ldots \models \alpha \}\]
\vskip-0.5ex
\noindent
A B\"uchi automaton $\mathcal{G}$ is good-for-MDP (GFM) if $\prb^{\max}_{\M}(\lang(\mathcal{G})) = \prb^{\max}_{\M \times \mathcal{G}}(\Pi_{acc})$ holds for all MDP $\M$~\cite{HahnPSSTW2020}.
 The inequality ``$\geq$'' holds for all automata\cite[Thm.~1]{KleinMBK2014}, but the other direction requires, intuitively, that a scheduler on $\M \times \mathcal{G}$ is able to safely resolve the nondeterminism of the automaton based on the prefix of the run.
 This is trivially satisfied by deterministic automata, but \emph{good-for-games} automata also have this property~\cite{KleinMBK2014}.
 Limit-deterministic B\"uchi automata are not GFM in general, for example, $\ld_{\mathcal{A}}$ may not be (see~\Cref{ex:singnotsuff}).

We fix an arbitrary MDP $\M$ and show that $\prb^{\max}_{\M}(\lang(\mathcal{A})) \leq \prb^{\max}_{\M \times \gfm_\mathcal{A}}(\Pi_{acc})$.
To this end we show that for any finite-memory scheduler $\mathfrak{S}$ on $\M$ we find a scheduler $\mathfrak{S'}$ on $\M \times \gfm_\mathcal{A}$ such that $\Pr^{\mathfrak{S}}_{\M}(\lang(\mathcal{A})) \leq \Pr^{\mathfrak{S}'}_{\M \times \gfm_\mathcal{A}}(\Pi_{acc})$.
The restriction to finite-memory schedulers is allowed because the maximal probability to satisfy an $\omega$-regular property is always attained by such a scheduler~\cite[Secs. 10.6.3 and 10.6.4]{BaierK2008}.
Let $\M_{\mathfrak{S}} \times \mathcal{D}$ be the product of the induced finite Markov chain $M_{\mathfrak{S}}$ with $\mathcal{D} = \bigotimes_{1 \leq i \leq m} \mathcal{D}_i$, where $\mathcal{D}_i = \determin\big(\removeFin(\splitTELAi{\mathcal{A}}{i})\big)$ and ``$\determin$'' is the GBA-determinization procedure from~\cite{ScheweV2012}, which makes $\mathcal{D}$ a deterministic Rabin automaton.
The scheduler $\mathfrak{S}'$ is constructed as follows.
It stays inside the initial component of $\M \times \gfm_{\mathcal{A}}$ and mimics the action chosen by $\mathfrak{S}$ until the corresponding path in $\M_{\mathfrak{S}} \times \mathcal{D}$ reaches an accepting bottom strongly connected component (BSCC) $B$.
This means that the transitions of $\mathcal{D}$ induced by $B$ satisfy one of the Rabin pairs.
The following lemma shows that in this case there exists a state in one of the breakpoint components to which $\mathfrak{S}'$ can safely move.

\begin{restatable}{lemma}{Goodimpliesacc}
  \label{lem:goodimpliesacc}
  Let $\mathfrak{s}$ be a state in an accepting BSCC $B$ of $\M_{\mathfrak{S}} \times \mathcal{D}$ and $\pi_1$ be a finite path that reaches $\mathfrak{s}$ from the initial state of $\M_{\mathfrak{S}} \times \mathcal{D}$.
  Then, there exists $1 \leq i \leq m$ and $Q' \subseteq \theta\bigl(I,L(\pi_1)\bigr)$ such that:
  \[\Pr_{\mathfrak{s}}(\{\pi \mid L(\pi) \text{ is accepted from } (Q',\varnothing,0) \text{ in } \breakp_i\bigl \}) = 1\]
\end{restatable}
The lemma does not hold if we restrict ourselves to singleton $\{q\} \subseteq \theta\bigl(I,L(\pi_1)\bigr)$ (see~\Cref{ex:singnotsuff}).
Hence, restricting $\gfmbridge$ to such transitions (as for $\ldbridge$, see~\Cref{def:gfmdef}) would not guarantee the GFM property.

\begin{example}
  \label{ex:singnotsuff}
  Consider the automaton $\mathcal{A}$ with states $\{\mathtt{a}_i\mathtt{b}_j \mid i,j \in \{1,2\}\} \cup \{\mathtt{b}_i\mathtt{a}_j \mid i,j \in \{1,2\}\}$, where $\mathtt{a}_i\mathtt{b}_j$ has transitions labeled by $a_i$ to $\mathtt{b}_j\mathtt{a}_1$ and $\mathtt{b}_j\mathtt{a}_2$.
  Transitions of states $\mathtt{b}_i\mathtt{a}_j$ are defined analogously, and all states in $\{\mathtt{a}_i\mathtt{b}_j \mid i,j \in \{1,2\}\}$ are initial (\Cref{subfig:singexaut} shows the transitions of $\mathtt{a}_1\mathtt{b}_1$).
  All transitions are accepting for a single B\"uchi condition, and hence $\lang(\mathcal{A}) = (\{a_i b_j \mid i,j \in \{1,2\}\})^{\omega}$.

  Consider the Markov chain $\M$ in~\Cref{subfig:singexmc} (transition probabilities are all $1/2$ and ommitted in the figure).
  Clearly, $\Pr_{\M}(\lang(\mathcal{A})) = 1$.
  \Cref{subfig:singexbreak} shows a part of the product of $\M$ with the breakpoint automaton $\breakp$ for $\mathcal{A}$ (\Cref{def:breakp}) starting from $\bigl(a_1,(\{\mathtt{a}_1\mathtt{b}_1\},\varnothing,0)\bigr)$.
  The state $\bigl(b_2,(\{\mathtt{b}_1\mathtt{a}_1,\mathtt{b}_1\mathtt{a}_2\},\varnothing,0)\bigr)$ is a trap state as $\mathtt{b}_1\mathtt{a}_1$ and $\mathtt{b}_1\mathtt{a}_2$ have no $b_2$-transition.
  Hence, $\bigl(a_1,(\{\mathtt{a}_1\mathtt{b}_1\},\varnothing,0)\bigr)$ generates an accepting path with probability at most $1/2$.
  This is true for all states $\bigl(s,(P',\varnothing,0)\bigr)$ of $\M \times \breakp$ where $P'$ is a singleton.
  But using $\ldbridge$ to connect initial and accepting components implies that any accepting path sees such a state.
  Hence, using $\ldbridge$ to define $\gfm_{\mathcal{A}}$ would not guarantee the GFM property.
\end{example}

\begin{figure}[tbp]
  \begin{subfigure}[t]{0.33\textwidth}
  \centering
  \includegraphics{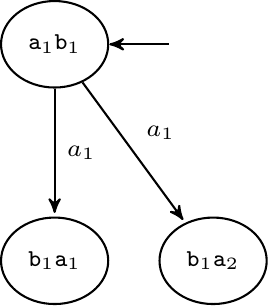}
  \caption{}
  \label{subfig:singexaut}
  \end{subfigure}\hfill
  \begin{subfigure}[t]{0.33\textwidth}
  \centering
  \includegraphics{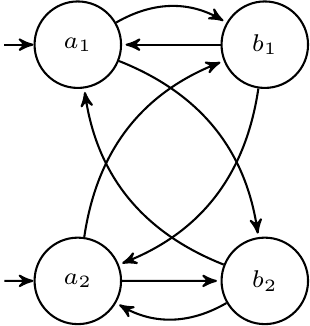}
  \caption{}
  \label{subfig:singexmc}
  \end{subfigure}\hfill
  \begin{subfigure}[t]{0.33\textwidth}
  \centering
  \includegraphics{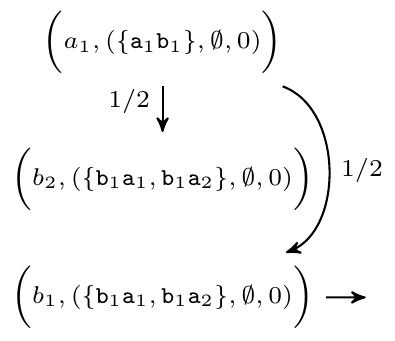}
  \caption{}
  \label{subfig:singexbreak}
  \end{subfigure}
  \caption{Restricting $\gfmbridge$ to transitions with endpoints of the form $(s, (\{q\},\varnothing,0))$ (similar to $\ldbridge$) would not guarantee the GFM property (see~\Cref{ex:singnotsuff}).}
  \label{fig:singletonsinsufficient}
\end{figure}
Using~\Cref{lem:goodimpliesacc} we can define $\S'$ such that the probability accepting paths under $\mathfrak{S}'$ in $\M \times \gfm_{\mathcal{A}}$ is at least as high as that of paths with label in $\lang(\mathcal{A})$ in $\M_{\mathfrak{S}}$.
This is the non-trivial direction of the GFM property.
\begin{restatable}{lemma}{nontrivinclusionGFM}
  \label{lem:nontrivinclusionGFM}
  For every finite-memory scheduler $\mathfrak{S}$ on $\M$, there exists a scheduler $\mathfrak{S}'$ on $\M \times \gfm_{\mathcal{A}}$ such that:
  \[\Pr_{\M \times \gfm_{\mathcal{A}}}^{\mathfrak{S}'}(\Pi_{acc}) \geq \Pr_{\M}^{\mathfrak{S}}(\lang(\mathcal{A}))\]
\end{restatable}
\begin{restatable}{proposition}{GFMproperty}
  \label{thm:gfmprop}
  The automaton $\gfm_{\mathcal{A}}$ is \emph{good-for-MDP}.
\end{restatable}
\noindent To compute $\prb_{\M}^{\max}(\lang(\mathcal{B}))$ one can translate $\mathcal{B}$ into an equivalent TELA $\mathcal{A}$ in DNF, then construct $\gfm_{\mathcal{A}}$ and finally compute $\prb^{\max}_{\M \times \gfm_{\mathcal{A}}}(\Pi_{\operatorname{acc}})$.
The automaton $\gfm_{\mathcal{A}}$ is single-exponential in the size of $\mathcal{B}$ by~\Cref{thm:GFMcorrectness}, and $\prb^{\max}_{\M \times \gfm_{\mathcal{A}}}(\Pi_{\operatorname{acc}})$ can be computed in polynomial time in the size of $\M \times \gfm_{\mathcal{A}}$\cite[Thm. 10.127]{BaierK2008}. %
\begin{theorem}
  \label{cor:quantsingleexp}
  Given a TELA $\mathcal{B}$ (not necessarily in DNF) and an MDP $\M$, the value $\prb^{\max}_{\M}(\lang(\mathcal{B}))$ can be computed in single-exponential time.
\end{theorem}

\section{Experimental evaluation}\label{sec:evaluation}
The product approach combines a sequence of deterministic automata using the disjunctive product.
We introduce the \emph{langcover heuristic}: the automata are ``added'' to the product one by one, but only if their language is not already subsumed by the automaton constructed so far.
This leads to substantially smaller automata in many cases, but is only efficient if checking inclusion for the considered automata types is efficent.
In our case this holds (the automata are deterministic with a disjunction of parity conditions as acceptance), but it is not the case for arbitrary deterministic TELA, or nondeterministic automata.

\paragraph{Implementation.}
We compare the following implementations of the constructions discussed above.\footnote{The source code and data of all experiments are available at~\cite{John2021}.}
\texttt{\spot{}} uses the TELA to GBA translator of \spot{}, simplifies (using \spot's \texttt{postprocessor} with preference \texttt{Small}) and degeneralizes the result and then determinizes using a version of Safra's algorithm \cite{Duret-LutzLFMRX2016, Redziejowski2012}.
The $\removeFin$ function that is used is an optimized version of~\Cref{def:removeFin}.
In \texttt{\approachOne{}}, \texttt{\approachTwo{}} and \texttt{\approachThree{}}, the first step is replaced by the corresonding TELA to GBA construction (using \spot's $\removeFin$).
The product approach (also implemented using the \spot{}-library) is called \texttt{product} and \texttt{product (no langcover)} (without the langcover heuristic).
The intermediate GBA are also simplified.
The construction $\ld_{\mathcal{A}}$ is implemented in \texttt{limit-det.}, using the \spot{}-library and parts of \seminator{}.
We compare it to \texttt{limit-det. via GBA}, which concatenates the TELA to GBA construction of \spot{} with the limit-determinization of \seminator{}.
Similarly, \texttt{good-for-MDP} and \texttt{good-for-MDP via GBA} are the construction $\gfm_{\mathcal{A}}$ applied to $\mathcal{A}$ directly, or to the GBA as constructed by \spot{}.
Both constructions \texttt{via GBA} are in the worst case double-exponential.
No post-processing is applied to any output automaton.

\paragraph{Experiments.}
Computations were performed on a computer with two Intel E5-2680 CPUs with 8 cores each at \(2.70\)\,GHz running Linux.
Each individual experiment was limited to a single core, $15$ GB of memory and 1200 seconds. 
We use versions 2.9.4 of \spot{} (configured to allow $256$ acceptance sets) and 2.0 of \seminator{}.

\begin{table}[btp]
	\caption{Evaluation of benchmarks \emph{random} and \emph{DNF}. 
		Columns ``states'', ``time'' and ``acceptance'' refer to the respective median values, where mem-/timeouts are counted as larger than the rest. Values in brackets refer to the subset of input automata for which at least one determinization needed more than $0.5$ seconds (447 (182) automata for benchmark \emph{random} (\emph{DNF})).}
	\label{tab:randomAutomata}
	\begin{center}
    \resizebox{\linewidth}{!}{
      \setlength{\tabcolsep}{4pt}
		\begin{tabular}{l l l l l l l l l}
			 & \multirow{2}{*}{algorithm} & \multirow{2}{*}{timeouts} & \multirow{2}{*}{memouts} & \multirow{2}{*}{states}  & \multirow{2}{*}{time} & \multirow{2}{*}{acceptance} & \multicolumn{2}{c}{intermediate GBA} \\ 
			 & & & & &  & & states & acceptance \\
			\hline
			\multirow{10}{*}{\rotatebox[origin=c]{90}{\emph{random}}} & \texttt{\spot{}} & 0.5\% & 9.9\%  & 3,414 (59,525) & $<1$ (1.5) & 10 (17) & 71 & 2 \\ 
			& \texttt{\approachOne} & 0.5\% & 15.2\% & 8,639 (291,263) & $<1$ (9.7) & 14 (24) & 109 & 2 \\ 
			& \texttt{\approachTwo} & 0.7\% & 17.8\% & 14,037 (522,758) & $<1$ (21.0) & 14 (24)& 119 & 2 \\ 
			& \texttt{\approachThree} & 1.6\% & 18.7\% & 15,859 (1,024,258)& $<1$ (40.2)& 14 (26) & 116 & 2 \\ 
			& \texttt{product} & 1.3\% & 7.9\% & 3,069 (43,965) & $<1$ (1.2) & 18 (29) \\ 
			& \texttt{product (no langcover)} & 0.7\% & 9.0\% & 3,857 (109,908) & $<1$ (1.1) & 24 (38) \\ \cline{2-9}
			& \texttt{limit-det.} &  0.0\% & 0.0\% & 778 (3,346) & $<1$ ($<1$) & 1 (1)\\
			& \texttt{limit-det. via GBA} & 1.6\% & 0.3\% & 463 (1,556) & $<1$ (1.6)& 1 (1) \\
			& \texttt{good-for-MDP} & {9.3\%} & {13.4\%} & {5,069} (192.558) & {2.0} (139.6) & 1 (1)\\
			& \texttt{good-for-MDP via GBA} & {5.5\%} & {44.0\%} & 71,200 (--) & 836.9 (--) & 1 (--) \\
			\hline
			\hline
			\multirow{2}{*}{\rotatebox[origin=c]{90}{\emph{DNF}}} & \texttt{\spot{}} & 0.4\% & 6.2\% & 5,980 (692,059) & $<1$ (18.3) &  11 (25) & 30 & 3\\
			& \texttt{product} & 0.0\% & 3.8\% & 2,596 (114,243) & $<1$ (4.6) & 13 (24)
		\end{tabular}}
	\end{center}
\end{table}

Our first benchmark set (called \emph{random}) consists of 1000 TELA with $4$ to $50$ states and $8$ sets of transitions $T_1, \ldots, T_8$ used to define the acceptance conditions.
They are generated using \spot's procedure \texttt{random\_graph()} by specifying probabilities such that: a triple $(q,a,q') \in Q \times \Sigma \times Q$ is included in the transition relation ($3/|Q|$) and such that a transition $t$ is included in a set $T_j$ ($0.2$). We use only transition systems that are nondeterministic.
The acceptance condition is generated randomly using \spot's procedure \texttt{acc\_code::random()}.
We transform the acceptance condition to DNF and keep those acceptance conditions whose lengths range between 2 and 21 and consist of at least two disjuncts.
To quantify the amount of nondeterminism, we divide the number of pairs of transitions of the form $(q,a,q_1),(q,a,q_2)$, with $q_1 \neq q_2$, of the automaton by its number of states. 

\Cref{tab:randomAutomata} shows that the \texttt{product} produces smallest deterministic automata overall.
\texttt{\spot{}} produces best results among the algorithms that go via a single GBA.
One reason for this is that after GBA-simplifications of \spot{}, the number of acceptance marks of the intermediate GBA are comparable.
Figure~\ref{fig:SpotVsProduct} (left) compares \texttt{\spot{}} and \texttt{product} and partitions the input automata according to acceptance complexity (measured in the size of their DNF) and amount of nondeterminism. 
Each subset of input automata is of roughly the same size (159-180) (see \Cref{tab:BenchmarksGroupedNondeterminism} in the appendix). 
The graph depicts the median of the ratio (\texttt{product} / \texttt{\spot}) for the measured values. For time- or memouts of \texttt{\spot} (\texttt{product}) we define the ratio as 0 ($\infty$). If both failed, the input is discarded. The number of time- and memouts grows with the amount of nondeterminism and reaches up to 42\%.
The approach \texttt{product} performs better for automata with more nondeterminism and complex acceptance conditions as the results have fewer states and the computation times are smaller compared to \texttt{\spot}. %

The limit-deterministic automata are generally much smaller than the deterministic ones, and \texttt{limit-det via GBA.} performs best in this category.
However, the construction $\ld_{\mathcal{A}}$ (\texttt{limit-det.}) resulted in fewer time- and memouts.

For GFM automata we see that computing $\gfm_{\mathcal{A}}$ directly, rather than first computing a GBA, yields much better results (\texttt{good-for-MDP} vs. \texttt{good-for-MDP via GBA}).
However, the GFM automata suffer from significantly more time- and memouts than the other approaches.
The automata sizes are comparable on average with \spot{}'s determinization (see~\Cref{fig:EvaluationPairwise} in the appendix).
Given their similarity to the pure limit-determinization constructions, and the fact that their acceptance condition is much simpler than for the deterministic automata, we believe that future work on optimizing this construction could make it a competitive alternative for probabilistic model checking using TELA.
\begin{figure}[tbp]
	\centering
	\begin{subfigure}{.58\textwidth}
		\centering
		\includegraphics[scale=0.7]{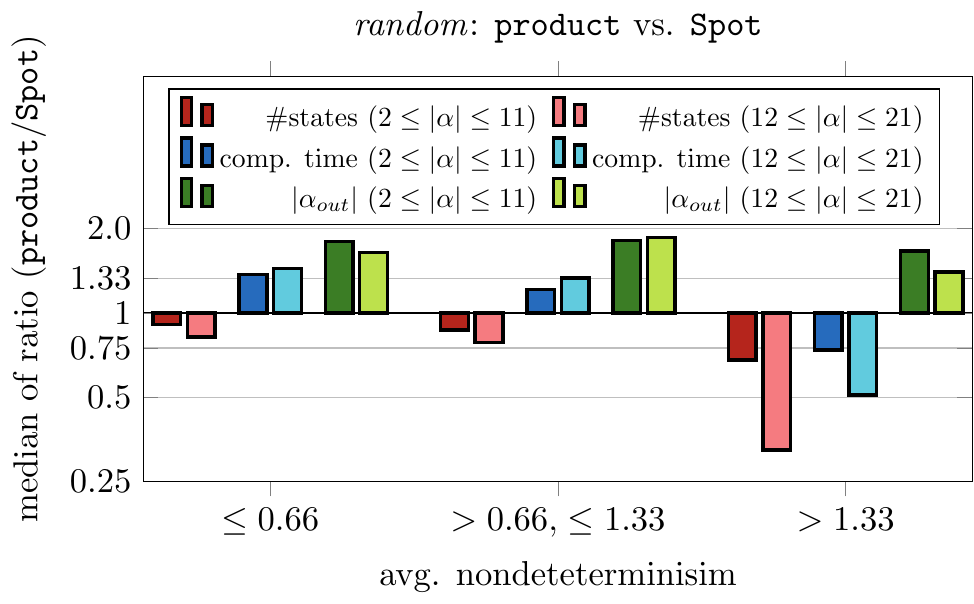}
	\end{subfigure}%
	\begin{subfigure}{.41\textwidth}
		\centering
		\includegraphics[scale=0.7]{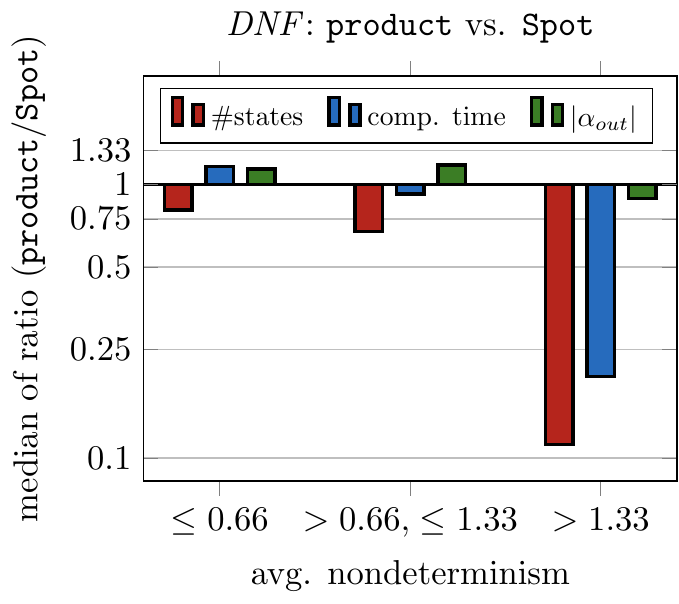}
	\end{subfigure}
	\caption{Comparison of \texttt{\spot} and \texttt{product}, with input automata grouped by the size of the DNF of their acceptance condition and the amount of nondeterminism.}
	\label{fig:SpotVsProduct}
\end{figure}

The second benchmark (called \emph{DNF}) consists of $500$ TELA constructed randomly as above, apart from the acceptance conditions.
They are in DNF with $2$-$3$ disjuncts, with $2$-$3$ $\infset$-atoms and $0$-$1$ $\finset$-atoms each (all different).
Such formulas tend to lead to larger CNF conditions, which benefits the new approaches.
\Cref{fig:SpotVsProduct} (right) shows the median ratio of automata sizes, computation times and acceptance sizes, grouped by the amount of nondeterminism. We do not consider different lengths of acceptance conditions because the subsets of input automata are already relatively small (140-193).
Again, \texttt{product} performs better for automata with more nondeterminism.
\section{Conclusion}
We have introduced several new approaches to determinize and limit-determinize automata under the Emerson-Lei acceptance condition.
The experimental evaluation shows that in particular the product approach performs very well.
Furthermore, we have shown that the complexity of limit-determinizing TELA is single-exponential (in contrast to the double-exponential blow-up for determinization).
One of our constructions produces limit-deterministic good-for-MDP automata, which can be used for quantitative probabilistic verification.

This work leads to several interesting questions.
The presented constructions would benefit from determinization procedures for GBA which trade a general acceptance condition (rather than Rabin or parity) for a more compact state-space of the output.
Similarly, translations from LTL to compact, nondeterministic TELA would allow them to be embedded into (probabilistic) model-checking tools for LTL (a first step in this direction is made in\cite{MajorBFSSZ2019}).
It would be interesting to study, in general, what properties can be naturally encoded directly into nondeterministic TELA.
Another open point is to evaluate the good-for-MDP automata in the context of probabilistic model checking in practice. \\[0.1cm]
\noindent \textbf{Acknowledgments.} We thank David Müller for suggesting to us the problem of determinizing Emerson Lei automata and many discussions on the topic.

\bibliographystyle{splncs04}

\begin{thebibliography}{10}
\providecommand{\url}[1]{\texttt{#1}}
\providecommand{\urlprefix}{URL }
\providecommand{\doi}[1]{https://doi.org/#1}

\bibitem{AllenEmersonL1987}
Allen~Emerson, E., Lei, C.L.: Modalities for model checking: Branching time
  logic strikes back. Science of Computer Programming  \textbf{8}(3),  275--306
  (Jun 1987). \doi{10.1016/0167-6423(87)90036-0}

\bibitem{BabiakBDKKMPS2015}
Babiak, T., Blahoudek, F., {Duret-Lutz}, A., Klein, J., K{\v r}et{\'i}nsk{\'y},
  J., M{\"u}ller, D., Parker, D., Strej{\v c}ek, J.: The {{Hanoi
  Omega}}-{{Automata Format}}. In: Kroening, D., P{\u a}s{\u a}reanu, C.S.
  (eds.) Computer {{Aided Verification}}. pp. 479--486. Lecture {{Notes}} in
  {{Computer Science}}, {Springer International Publishing}, {Cham} (2015).
  \doi{10.1007/978-3-319-21690-4\_31}

\bibitem{BaierBDKMS2019}
Baier, C., Blahoudek, F., {Duret-Lutz}, A., Klein, J., M{\"u}ller, D., Strej{\v
  c}ek, J.: Generic {{Emptiness Check}} for {{Fun}} and {{Profit}}. In: Chen,
  Y.F., Cheng, C.H., Esparza, J. (eds.) Automated {{Technology}} for
  {{Verification}} and {{Analysis}}. pp. 445--461. Lecture {{Notes}} in
  {{Computer Science}}, {Springer International Publishing}, {Cham} (2019).
  \doi{10.1007/978-3-030-31784-3\_26}

\bibitem{BaierK2008}
Baier, C., Katoen, J.P.: Principles of {{Model Checking}} ({{Representation}}
  and {{Mind Series}}). {The MIT Press} (2008)

\bibitem{Ben-Ari2008}
{Ben-Ari}, M.: Principles of the {{Spin Model Checker}}. {Springer-Verlag},
  {London} (2008). \doi{10.1007/978-1-84628-770-1}

\bibitem{Blahoudek2018}
Blahoudek, F.: Automata for {{Formal Methods}}: {{Little Steps Towards
  Perfection}}. Ph.D. thesis, Masaryk University, Faculty of Informatics
  (2018), \url{https://is.muni.cz/th/gwriw/?lang=en}

\bibitem{BlahoudekDKKS2017}
Blahoudek, F., {Duret-Lutz}, A., Kloko{\v c}ka, M., K{\v r}et{\'i}nsk{\'y}, M.,
  Strej{\v c}ek, J.: Seminator: {{A Tool}} for {{Semi}}-{{Determinization}} of
  {{Omega}}-{{Automata}}. In: {{EPiC Series}} in {{Computing}}. vol.~46, pp.
  356--367. {EasyChair} (May 2017). \doi{10.29007/k5nl}

\bibitem{BlahoudekMS2019a}
Blahoudek, F., Major, J., Strej{\v c}ek, J.: {{LTL}} to {{Smaller Self}}-{{Loop
  Alternating Automata}} and {{Back}}. In: Hierons, R.M., Mosbah, M. (eds.)
  Theoretical {{Aspects}} of {{Computing}} \textendash{} {{ICTAC}} 2019. pp.
  152--171. Lecture {{Notes}} in {{Computer Science}}, {Springer International
  Publishing}, {Cham} (2019). \doi{10.1007/978-3-030-32505-3\_10}

\bibitem{BloemenDv2019}
Bloemen, V., {Duret-Lutz}, A., {van de Pol}, J.: Model checking with
  generalized {{Rabin}} and {{Fin}}-less automata. International Journal on
  Software Tools for Technology Transfer  \textbf{21}(3),  307--324 (Jun 2019).
  \doi{10.1007/s10009-019-00508-4}

\bibitem{Boker2018}
Boker, U.: Why {{These Automata Types}}? In: {{EPiC Series}} in {{Computing}}.
  vol.~57, pp. 143--163. {EasyChair} (Oct 2018). \doi{10.29007/c3bj}

\bibitem{ChatterjeeGK2013}
Chatterjee, K., Gaiser, A., K{\v r}et{\'i}nsk{\'y}, J.: Automata with
  {{Generalized Rabin Pairs}} for {{Probabilistic Model Checking}} and {{LTL
  Synthesis}}. In: Sharygina, N., Veith, H. (eds.) Computer {{Aided
  Verification}}. pp. 559--575. Lecture {{Notes}} in {{Computer Science}},
  {Springer}, {Berlin, Heidelberg} (2013). \doi{10.1007/978-3-642-39799-8\_37}

\bibitem{CourcoubetisY1995}
Courcoubetis, C., Yannakakis, M.: The complexity of probabilistic verification.
  Journal of the ACM  \textbf{42}(4),  857--907 (Jul 1995).
  \doi{10.1145/210332.210339}

\bibitem{Couvreur1999}
Couvreur, J.M.: On-the-fly {{Verification}} of {{Linear Temporal Logic}}. In:
  Wing, J.M., Woodcock, J., Davies, J. (eds.) {{FM}}'99 \textemdash{} {{Formal
  Methods}}. pp. 253--271. Lecture {{Notes}} in {{Computer Science}},
  {Springer}, {Berlin, Heidelberg} (1999). \doi{10.1007/3-540-48119-2\_16}

\bibitem{DuretLutz2017}
Duret-Lutz, A.: Contributions to LTL and $\omega$-Automata for Model Checking.
  Habilitation thesis, Universit{\'e} Pierre et Marie Curie (2017)

\bibitem{Duret-LutzLFMRX2016}
{Duret-Lutz}, A., Lewkowicz, A., Fauchille, A., Michaud, T., Renault, {\'E}.,
  Xu, L.: Spot 2.0 \textemdash{} {{A Framework}} for {{LTL}} and
  {$\omega$}-{{Automata Manipulation}}. In: Artho, C., Legay, A., Peled, D.
  (eds.) Automated {{Technology}} for {{Verification}} and {{Analysis}}. pp.
  122--129. Lecture {{Notes}} in {{Computer Science}}, {Springer International
  Publishing}, {Cham} (2016). \doi{10.1007/978-3-319-46520-3\_8}

\bibitem{DuretLutzJMZ2009}
Duret-Lutz, A., Poitrenaud, D., Couvreur, J.M.: On-the-fly {Emptiness} {Check}
  of {Transition}-{Based} {Streett} {Automata}. In: Automated Technology for
  Verification and Analysis (ATVA). {{LNCS}}, Springer (2009)

\bibitem{EsparzaKS2018}
Esparza, J., K{\v r}et{\'i}nsk{\'y}, J., Sickert, S.: One {{Theorem}} to {{Rule
  Them All}}: {{A Unified Translation}} of {{LTL}} into $\omega$-{{Automata}}.
  In: Proceedings of the 33rd {{Annual ACM}}/{{IEEE Symposium}} on {{Logic}} in
  {{Computer Science}}. pp. 384--393. {{LICS}} '18, {Association for Computing
  Machinery}, {New York, NY, USA} (Jul 2018). \doi{10.1145/3209108.3209161}

\bibitem{GiannakopoulouL2002}
Giannakopoulou, D., Lerda, F.: From {{States}} to {{Transitions}}: {{Improving
  Translation}} of {{LTL Formulae}} to {{B\"uchi Automata}}. In: Peled, D.A.,
  Vardi, M.Y. (eds.) Formal {{Techniques}} for {{Networked}} and {{Distributed
  Sytems}} \textemdash{} {{FORTE}} 2002. pp. 308--326. Lecture {{Notes}} in
  {{Computer Science}}, {Springer}, {Berlin, Heidelberg} (2002).
  \doi{10.1007/3-540-36135-9\_20}

\bibitem{HahnLSTZ2015}
Hahn, E.M., Li, G., Schewe, S., Turrini, A., Zhang, L.: Lazy {{Probabilistic
  Model Checking}} without {{Determinisation}}. arXiv:1311.2928 [cs]  (Apr
  2015), \url{http://arxiv.org/abs/1311.2928}

\bibitem{HahnPSSTW2020}
Hahn, E.M., Perez, M., Schewe, S., Somenzi, F., Trivedi, A., Wojtczak, D.:
  Good-for-{{MDPs Automata}} for {{Probabilistic Analysis}} and {{Reinforcement
  Learning}}. In: Biere, A., Parker, D. (eds.) Tools and {{Algorithms}} for the
  {{Construction}} and {{Analysis}} of {{Systems}}. pp. 306--323. Lecture
  {{Notes}} in {{Computer Science}}, {Springer International Publishing},
  {Cham} (2020). \doi{10.1007/978-3-030-45190-5\_17}

\bibitem{John2021}
John, T., Jantsch, S., Baier, C., Klüppelholz, S.: {Determinization and
  Limit-determinization of Emerson-Lei Automata - Supplementary material
  (ATVA'21)}  (2021). \doi{10.6084/m9.figshare.14838654.v2}

\bibitem{KleinMBK2014}
Klein, J., M{\"u}ller, D., Baier, C., Kl{\"u}ppelholz, S.: Are
  {{Good}}-for-{{Games Automata Good}} for {{Probabilistic Model Checking}}?
  In: Dediu, A.H., {Mart{\'i}n-Vide}, C., {Sierra-Rodr{\'i}guez}, J.L., Truthe,
  B. (eds.) Language and {{Automata Theory}} and {{Applications}}. pp.
  453--465. Lecture {{Notes}} in {{Computer Science}}, {Springer International
  Publishing}, {Cham} (2014). \doi{10.1007/978-3-319-04921-2\_37}

\bibitem{KretinskyMS2018}
K{\v r}et{\'i}nsk{\'y}, J., Meggendorfer, T., Sickert, S.: Owl: {{A Library}}
  for \$\$\textbackslash omega \$\$-{{Words}}, {{Automata}}, and {{LTL}}. In:
  Lahiri, S.K., Wang, C. (eds.) Automated {{Technology}} for {{Verification}}
  and {{Analysis}}. pp. 543--550. Lecture {{Notes}} in {{Computer Science}},
  {Springer International Publishing}, {Cham} (2018).
  \doi{10.1007/978-3-030-01090-4\_34}

\bibitem{KwiatkowskaNP2011}
Kwiatkowska, M., Norman, G., Parker, D.: {{PRISM}} 4.0: {{Verification}} of
  {{Probabilistic Real}}-{{Time Systems}}. In: Gopalakrishnan, G., Qadeer, S.
  (eds.) Computer {{Aided Verification}}. pp. 585--591. Lecture {{Notes}} in
  {{Computer Science}}, {Springer}, {Berlin, Heidelberg} (2011).
  \doi{10.1007/978-3-642-22110-1\_47}

\bibitem{LodingP2019a}
L{\"o}ding, C., Pirogov, A.: Determinization of {{B\"uchi Automata}}:
  {{Unifying}} the {{Approaches}} of {{Safra}} and {{Muller}}-{{Schupp}}. In:
  Baier, C., Chatzigiannakis, I., Flocchini, P., Leonardi, S. (eds.) 46th
  {{International Colloquium}} on {{Automata}}, {{Languages}}, and
  {{Programming}} ({{ICALP}} 2019). Leibniz {{International Proceedings}} in
  {{Informatics}} ({{LIPIcs}}), vol.~132, pp. 120:1--120:13. {Schloss
  Dagstuhl\textendash Leibniz-Zentrum fuer Informatik}, {Dagstuhl, Germany}
  (2019). \doi{10.4230/LIPIcs.ICALP.2019.120}

\bibitem{MajorBFSSZ2019}
Major, J., Blahoudek, F., Strej{\v{c}}ek, J., Sasar{\'a}kov{\'a}, M.,
  Zbon{\v{c}}{\'a}kov{\'a}, T.: ltl3tela: {LTL} to {Small} {Deterministic} or
  {Nondeterministic} {Emerson}-{Lei} {Automata}. In: Automated Technology for
  Verification and Analysis (ATVA) (2019)

\bibitem{MiyanoH1984}
Miyano, S., Hayashi, T.: Alternating finite automata on {$\omega$}-words.
  Theoretical Computer Science  \textbf{32}(3),  321--330 (Jan 1984).
  \doi{10.1016/0304-3975(84)90049-5}

\bibitem{Mueller2019}
M{\"u}ller, D.: {Alternative Automata-based Approaches to Probabilistic Model
  Checking}. Ph.D. thesis, Technische Universit\"at Dresden (Nov 2019),
  \url{https://tud.qucosa.de/landing-page/?tx_dlf[id]=https%3A%2F%2Ftud.qucosa.de%2Fapi%2Fqucosa%253A36100%2Fmets}

\bibitem{MuellerS2017}
M{\"{u}}ller, D., Sickert, S.: {LTL} to {Deterministic} {Emerson}-{Lei}
  {Automata}. In: Games, Automata, Logics and Formal Verification (GandALF).
  {EPTCS} (2017)

\bibitem{MullerS1995}
Muller, D.E., Schupp, P.E.: Simulating alternating tree automata by
  nondeterministic automata: {{New}} results and new proofs of the theorems of
  {{Rabin}}, {{McNaughton}} and {{Safra}}. Theoretical Computer Science
  \textbf{141}(1),  69--107 (Apr 1995). \doi{10.1016/0304-3975(94)00214-4}

\bibitem{PnueliR1989}
Pnueli, A., Rosner, R.: On the synthesis of a reactive module. In: Proceedings
  of the 16th {{ACM SIGPLAN}}-{{SIGACT}} Symposium on {{Principles}} of
  Programming Languages. pp. 179--190. {{POPL}} '89, {Association for Computing
  Machinery}, {New York, NY, USA} (Jan 1989). \doi{10.1145/75277.75293}

\bibitem{Redziejowski2012}
Redziejowski, R.R.: An {{Improved Construction}} of {{Deterministic
  Omega}}-automaton {{Using Derivatives}}. Fundamenta Informaticae
  \textbf{119}(3-4),  393--406 (Jan 2012). \doi{10.3233/FI-2012-744}

\bibitem{RenkinDP2020}
Renkin, F., {Duret-Lutz}, A., Pommellet, A.: Practical ``{{Paritizing}}'' of
  {{Emerson}}-{{Lei Automata}}. In: Hung, D.V., Sokolsky, O. (eds.) Automated
  {{Technology}} for {{Verification}} and {{Analysis}}. pp. 127--143. Lecture
  {{Notes}} in {{Computer Science}}, {Springer International Publishing},
  {Cham} (2020). \doi{10.1007/978-3-030-59152-6\_7}

\bibitem{SafraV1989}
Safra, S., Vardi, M.Y.: On omega-automata and temporal logic. In: Proceedings
  of the Twenty-First Annual {{ACM}} Symposium on {{Theory}} of Computing. pp.
  127--137. {{STOC}} '89, {Association for Computing Machinery}, {New York, NY,
  USA} (Feb 1989). \doi{10.1145/73007.73019}

\bibitem{Safra1989}
Safra, S.: Complexity of {{Automata}} on {{Infinite Objects}}. Ph.D. thesis,
  Weizmann Institute of Science, {Rehovot, Israel} (1989)

\bibitem{ScheweV2012}
Schewe, S., Varghese, T.: Tight {{Bounds}} for the {{Determinisation}} and
  {{Complementation}} of {{Generalised B\"uchi Automata}}. In: Chakraborty, S.,
  Mukund, M. (eds.) Automated {{Technology}} for {{Verification}} and
  {{Analysis}}. pp. 42--56. Lecture {{Notes}} in {{Computer Science}},
  {Springer}, {Berlin, Heidelberg} (2012). \doi{10.1007/978-3-642-33386-6\_5}

\bibitem{SickertEJK2016}
Sickert, S., Esparza, J., Jaax, S., K{\v r}et{\'i}nsk{\'y}, J.:
  Limit-{{Deterministic B\"uchi Automata}} for {{Linear Temporal Logic}}. In:
  Chaudhuri, S., Farzan, A. (eds.) Computer {{Aided Verification}}. pp.
  312--332. Lecture {{Notes}} in {{Computer Science}}, {Springer International
  Publishing}, {Cham} (2016). \doi{10.1007/978-3-319-41540-6\_17}

\bibitem{Vardi1985a}
Vardi, M.Y.: Automatic verification of probabilistic concurrent finite state
  programs. In: 26th {{Annual Symposium}} on {{Foundations}} of {{Computer
  Science}} (Sfcs 1985). pp. 327--338 (Oct 1985). \doi{10.1109/SFCS.1985.12}

\end{thebibliography}

\appendix

\section{Proofs for \Cref{sec:viagba}}

\splitConj*
\begin{proof}
  ``$\subseteq$'': Let $\rho$ be a run of $\mathcal{A}$ for $w$ such that $\rho \models \alpha$.
  Then, it follows that $\rho \models \alpha_i$ for some $1 \leq i \leq m$.
  But then $\rho$ is also an accepting run of $\splitTELAi{\mathcal{A}}{i}$ and hence $w \in \lang(\splitTELAi{\mathcal{A}}{i})$.

  ``$\supseteq$'': Let $\rho$ be an accepting run of $\splitTELAi{\mathcal{A}}{i}$ for some $1 \leq i \leq m$.
  Then $\rho \models \alpha_i$, and hence $\rho \models \alpha$.
  It follows that $\rho$ is an accepting run of $\mathcal{A}$.
\end{proof}

\begin{figure}
  \begin{subfigure}[t]{0.33\textwidth}
  \centering
  \includegraphics{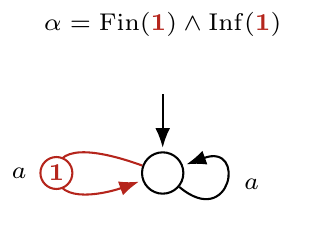}
  \caption{$\mathcal{A}$}
  \label{subfig:splitconj1}
  \end{subfigure}\hfill
  \begin{subfigure}[t]{0.33\textwidth}
  \centering
  \includegraphics{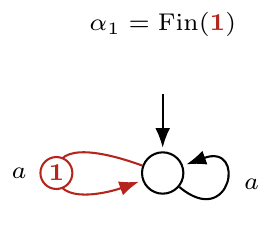}
  \caption{$\mathcal{A}_1$}
  \label{subfig:splitconj1}
  \end{subfigure}\hfill
  \begin{subfigure}[t]{0.33\textwidth}
  \centering
  \includegraphics{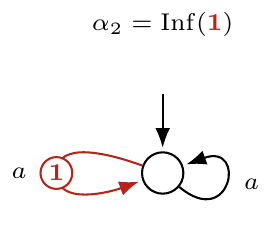}
  \caption{$\mathcal{A}_2$}
  \label{subfig:splitconj1}
  \end{subfigure}\hfill
  \caption{An example showing that the analogue of~\Cref{lem:splitAcceptance} for conjunction and intersection is not correct. In other words, the intersection of languages of automata one gets by splitting along a top-level conjunction is not necessarily the language of the original automaton. In this example we have: $\lang(\mathcal{A}) = \varnothing$ and $\lang(\mathcal{A}_1) = \lang(\mathcal{A}_2) = \{a^{\omega}\}$. But then $\lang(\mathcal{A}) \neq \lang(\mathcal{A}_1) \cap \lang(\mathcal{A}_2)$.}
  \label{fig:splitconj}
\end{figure}

\begin{proposition}
  Let $\mathcal{A}_i = (Q_i, \Sigma, \delta_i, I_i, \alpha_i)$, with $i \in \{0,1\}$, be two complete TELA with disjoint state-spaces.
  The following statements hold:
  \begin{enumerate}
  \item $\lang(\mathcal{A}_0) \cup \lang(\mathcal{A}_1) = \lang(\mathcal{A}_0 \oplus \mathcal{A}_1)$
  \item Assume that $\alpha_i = \infset(T_1^i) \wedge \ldots \wedge \infset(T_k^i)$, for $i \in \{0,1\}$. Then we have: $\lang(\mathcal{A}_0) \cup \lang(\mathcal{A}_1) = \lang(\mathcal{A}_0 \GBAsum \mathcal{A}_1)$
  \item $\lang(\mathcal{A}_0) \cup \lang(\mathcal{A}_1) = \lang(\mathcal{A}_0 \otimes \mathcal{A}_1)$
  \end{enumerate}
\end{proposition}
\begin{proof}
  1.) This is clear as any accepting run of $\mathcal{A}_i$ (for $i \in \{0,1\}$) can be mapped directly to an accepting run of $\mathcal{A}_0 \oplus \mathcal{A}_1$, and an accepting run $\mathcal{A}_0 \oplus \mathcal{A}_1$ corresponds to an accepting run of one of the automata $\mathcal{A}_0,\mathcal{A}_1$.

  2.) ``$\subseteq$'': Let $\rho$ be an accepting run of $\mathcal{A}_0$ (w.l.o.g).
  Then $\inf(\rho) \cap T_j^0 \neq \varnothing$ holds for all $1 \leq j \leq k$, and hence also $\inf(\rho) \cap (T_j^0 \cup T_j^{1}) \neq \varnothing$.
  But then $\rho$ is also an accepting run of $\mathcal{A}_0 \GBAsum \mathcal{A}_1$.

  ``$\supseteq$'': Let $\rho$ be an accepting run of $\mathcal{A}_0 \GBAsum \mathcal{A}_1$.
  Then $\inf(\rho) \cap (T_j^0 \cup T_j^{1})$ holds for all $1 \leq j \leq k$.
  Recall that $\delta_0 \cap \delta_1 = \varnothing$.
  We have either $\inf(\rho) \subseteq \delta_0$ or $\inf(\rho) \subseteq \delta_1$ and $T_j^i \subseteq \delta_i$ for all $1 \leq j \leq k$ and $i \in \{0,1\}$.
  As a consequence, there exists $i \in \{0,1\}$ such that $\inf(\rho) \cap T_j^i$ for all $1 \leq j \leq k$.
  But then $\rho$ is an accepting run of $\mathcal{A}_i$.
  
  3.) ``$\subseteq$'': Let $\rho = q_0 q_1 \ldots$ be an accepting run of $\mathcal{A}_0$ (w.l.o.g.) for $w$.
  As $\mathcal{A}_1$ is complete, we find a run $\rho^*$ of $\mathcal{A}_0 \otimes \mathcal{A}_1$ for $w$ such that $\rho^* = (q_0,q_0') (q_1,q_1') \ldots$, where $q_0' q_1' \ldots$ is a run of $\mathcal{A}_1$ for $w$.
  As $\rho \models \alpha_0$, it follows that $\rho^* \models \lifted{}{\alpha_0}$ and hence $\rho^*$ is an accepting run of $\mathcal{A}_0 \otimes \mathcal{A}_1$.

  ``$\supseteq$'': Let $\rho^* = (q_0,q_0') (q_1,q_1') \ldots$ be an accepting run of $\mathcal{A}_0 \otimes \mathcal{A}_1$ for $w$ and assume, w.l.o.g., that $\rho^* \models \lifted{}{\alpha_0}$.
  Then $q_0 q_1 $, as a run for $w$, models $\alpha_0$ and hence it is an accepting run of $\mathcal{A}_0$.
\end{proof}

\begin{figure}
	\begin{subfigure}[t]{0.33\textwidth}
		\centering
		\includegraphics{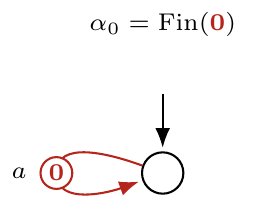}
		\caption{$\mathcal{A}_0$}
	\end{subfigure}\hfill
	\begin{subfigure}[t]{0.33\textwidth}
		\centering
		\includegraphics{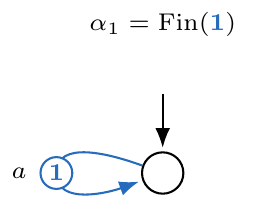}
		\caption{$\mathcal{A}_1$}
	\end{subfigure}\hfill
	\begin{subfigure}[t]{0.33\textwidth}
		\centering
		\includegraphics{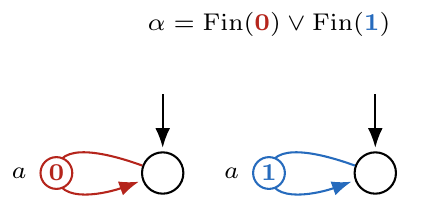}
		\caption{$ $}
		\label{subfig:Fin0v1}
	\end{subfigure}\hfill
	\caption{An example showing why we need to add the conditions $\infset(\delta_0)$ and $\infset(\delta_1)$ when constructing the acceptance condition of $\mathcal{A}_0 \oplus \mathcal{A}_1$. The TELA $\mathcal{A}_0$ and $\mathcal{A}_1$ both accept the empty language. If we unify the state spaces and disjunct the acceptance conditions we get the TELA in Figure~\ref{subfig:Fin0v1}. However, this automaton accepts the language $a^\omega \neq \lang(\mathcal{A}_0) \cup \lang(\mathcal{A}_1)$.}
	\label{fig:SumWithoutInf}
\end{figure}

\removeFinGba*
\begin{proof}
	
	We first prove that $\lang(\mathcal{A}) = \lang(\removeFin(\mathcal{A}))$.
	
	``$\subseteq$'': Let $\rho$ be an accepting run of $\mathcal{A}$ for $u$, and assume that $\rho \models \finset(T^i_0) \land \bigwedge_{1 \leq j \leq k_i} \infset(T^i_j)$ (the $i$'th disjunct of acceptance condition $\alpha$, which is in DNF).
  It follows that there exists a position $K$ after which $\rho$ sees no transitions in $T_0^{i}$.
  We construct an accepting run $\rho'$ of $\removeFin(\mathcal{A})$ for $u$: for the first $K$ positions, it copies $\rho$ in the main copy $Q$.
  Then, it moves to $Q_i$, and continues to simulate the moves of $\rho$.
  From that fact that $\rho$ sees infinitely many transitions in each set $T_j^i$, it follows that $\rho'$ models $\phi_i$ and hence is accepting.
		
  ``$\supseteq$'': Let $\rho'$ be an accepting run of $\removeFin(\mathcal{A})$ for $u$, i.e. $\rho' \models \phi_i$ for some $1 \leq i \leq m$.
  It follows that $\rho'$ eventually moves to $Q_i$, because $\phi_i$ contains at least one $\infset$-atom whose all  transitions are fully included in $Q_i$.
  As no other copy is reachable from $Q_i$, $\rho'$ stays in $Q_i$ therafter.
  Recall that copy $Q_i$ does not contain any transition in $U_0^{i}$.
  Projecting the part of $\rho'$ which is in $Q_i$ onto the corresponding transitions in $Q$ yields a run $\rho$ for $u$ in $\mathcal{A}$.
  From the definition of $\phi_i$ and the fact that $\rho'$ models $\phi_i$ it follows directly that $\rho$ models $\alpha_i$, and hence $\rho$ is accepting.
  
	We now show that $\lang(\removeFin(\mathcal{A})) = \lang(\removeFinGBA(\mathcal{A}))$ by proving that for all runs $\rho$ in $(\delta')^{\omega}$, the following equivalence holds: $\rho \models \alpha' \iff \rho \models \alpha''$.
  We recall that $\alpha' = \bigvee_{1 \leq i \leq m} \phi_i$, whith $\phi_i = \bigwedge_{1 \leq j \leq k_i} \infset(U^i_j)$ and $\alpha'' = \bigwedge_{1\leq j \leq k} \infset(U_j^1 \cup \ldots \cup U_j^m)$, with $k = \max_{i} k_i$ and $U_j^i = \delta_i$ if $k_i < j \leq k$.
	
  First, assume that $\rho \models \alpha'$, which implies that there exists an $1 \leq i \leq m$ such that $\rho \models \infset(U_j^i)$ for all $1 \leq j \leq k_i$.
  It follows that $\inf(\rho) \subseteq \delta_i$ and hence also $\rho \models U_j^i$ for all $k_i < j \leq k$.
  But then clearly $\rho \models \infset(U_j^1 \cup \ldots \cup U_j^m)$ for all $1 \leq j \leq k$ and hence $\rho \models \alpha''$.
  
	Now assume that $\rho \models \alpha''$.
  Then, in particular, $\rho \models \infset(U_1^1 \cup \ldots \cup U_1^m)$, which implies that $\rho \models \infset(U_1^i)$  for some $1 \leq i \leq m$.
  By construction of $\removeFin(\mathcal{A})$ it follows that $\inf(\rho) \subseteq \delta_i$ and $\inf(\rho) \cap \delta_{i'} = \varnothing$, for $i' \neq i$.
  It follows that $\inf(\rho) \cap U_j^{i'} = \varnothing$ for all $1 \leq j \leq k$ and $i' \neq i$.
  The only possibility for $\rho$ to satisfy all conjuncts in $\alpha''$ is to satisfy $\inf(\rho) \cap \infset(U_j^i) \neq \varnothing$ for all $1 \leq j \leq k$.
  But then it follows that $\rho \models \phi_i$ and hence $\rho \models \alpha'$.
\end{proof}

\section{Proofs for \Cref{sec:limitdet}}

We first show that there is a ``canonical'' partition $Q_N,Q_D$ for limit-determinism, which will be used in~\Cref{prop:limitdetcoml,prop:limitdetqual}.
Let $\mathcal{A} = (Q,I,\Sigma,\delta,\alpha)$ be a fixed TELA.
We say that a state $q \in Q$ is deterministic if it has at most one successor in $\delta$ for each symbol, and define:
\[Q_D^* = \{q \in Q \mid \text{ all states } q' \text{ reachable from } q \text{ are deterministic} \}\]
and $Q_N^* = Q  \setminus Q_D^*$.
This partition can be computed in polynomial time using an SCC analysis of $\mathcal{A}$.
\begin{lemma}
  \label{lem:canonicpart}
  $\mathcal{A}$ is limit-deterministic iff the partition $Q_D^*,Q_N^*$ satisfies conditions 1-3 of~\Cref{def:limitdet}.
\end{lemma}
\begin{proof}
  The direction from right to left is immediate.
  We now show that if there exists a partition $Q_D,Q_N$ which satisfies 1-3, then so does $Q_D^*, Q_N^*$.
  By construction, $Q_D^*, Q_N^*$ satisfy conditions 1. and 2.
  As every state in $Q_D$ is only allowed to reach deterministic states, it follows that $Q_D \subseteq Q_D^*$ and hence $Q_N^* \subseteq Q_N$.
  But then it follows from the fact that every accepting run of $\mathcal{A}$ eventually leaves $Q_N$ forever, that it also eventually leaves $Q_N^*$ forever.
  This shows that the partition $Q_D^*, Q_N^*$ satisfies condition 3.
\end{proof}

\LimitDetComplexity*
\begin{proof}
  As we have seen above, checking whether $\mathcal{A}$ is limit-deterministic amounts to checking whether the partition $Q_D^*, Q_N^*$ satisfies conditions 1-3 of~\Cref{def:limitdet}.
  Conditions 1-2 hold by construction.
  To check condition 3 one can construct an automaton $\mathcal{N}$ which contains all states and transitions of $Q_N$ and replaces all transitions out of $Q_N$ by a transition to a rejecting trap state (this may need a simple rewrite of the acceptance condition).
  Then, condition 3 holds if and only if $\mathcal{N}$ accepts the empty language, which can be checked in coNP for TELA (this follows from NP-completeness of the non-emptiness problem\cite[Thm. 4.7]{AllenEmersonL1987}).

  For coNP-hardness, we reduce from the emptiness problem of TELA, which is coNP-hard.
  First observe that an automaton $\mathcal{A}$ in which $Q_D^*$ is empty is limit-deterministic if and only if $\lang({\cal A}) = \varnothing$.
  So it suffices to translate an arbitrary TELA $\mathcal{B}$ into a TELA $\mathcal{A}$ in which $Q_D^*$ is empty, and such that $\lang({\cal B}) = \varnothing$ iff $\lang({\cal A}) = \varnothing$.
  To do this we add a nondeterministic SCC to $\mathcal{B}$, a transition from each original state in $\mathcal{B}$ to that SCC, and make sure using the acceptance condition that no run which gets trapped in the new SCC is accepting. 
\end{proof}

\LimitDetApprOne*
\begin{proof}
  The choice of initial component is the only nondeterminism of the resulting automaton.
  Hence we can take $Q_D$ to be the entire state-set and $Q_N = \varnothing$.
\end{proof}

\LDTELAqualitative*
\begin{proof}
  NP-hardness of the problem for general limit-det of deterministic TELA follows directly from NP-hardness for deterministic TELA (see~\cite[Thm. 5.13]{Mueller2019}).
  For the upper-bounds we use the fact discussed above that if $\mathcal{A}$ is limit-deterministic, then we can compute a partition $Q_D,Q_N$ satisfying conditions 1-3 of~\Cref{def:limitdet} in polynomial time.

  \textbf{In NP for limit-det. TELA.}
  We now show that the problem of computing $\prb_{\M}^{\max}(\lang(\mathcal{A})) > 0$ is in NP if $\mathcal{A}$ is limit-deterministic.
  Let $Q_N,Q_D$ be a partition satisfying conditions 1-3 of~\Cref{def:limitdet} and let $\M \times \mathcal{A}$ be the product-MDP defined in~\Cref{def:product}.
  Recall that an \emph{end-component} of an MDP is a non-empty subset of its states $S'$ together with a subset of the enabled actions $T(s) \subseteq \Act(s)$ for each state $s \in S'$ such that the underlying graph is strongly connected and closed under probabilistic transitions~\cite[Definition 10.117]{BaierK2008}.
  An end-component $\cal E$ of $\M \times \mathcal{A}$ naturally induces a set of transitions of $\mathcal{A}$, which we denote by $\mathtt{atrans}({\cal E})$.
  
  \textbf{Claim.} $\prb^{\max}_{\M}(\lang(\mathcal{A})) > 0$ holds iff there exists a reachable end-component $\cal E$ of $\M \times \mathcal{A}$ such that $\mathtt{atrans}({\cal E}) \models \alpha$.

  ``$\Longleftarrow$'': Let $(s,q)$ be a state contained in such an end-component $\cal E$.
  As $\cal E$ is assumed to be reachable, there exists a finite path $\pi = s_0 \alpha_0 s_1 \alpha_1 \ldots s_n \alpha_n s$ through $\M$ and a corresponding path $q_0 \xrightarrow{L(s_0 \ldots s_{n})} q$ through $\mathcal{A}$.
  We start defining a scheduler $\mathfrak{S}_1$ on $\M$ such that it chooses action $\alpha_i$ for all prefixes of length $i$, if $0 \leq i \leq n$ and $\alpha_i$ is enabled.
  Clearly, $\mathfrak{S}_1$ achieves a positive probability to realize the prefix $\pi$.
  As $\cal E$ is an end-component, we can construct a scheduler $\mathfrak{S}$ on $M \times \mathcal{A}$ from state $(s,q)$ such that the set of transitions visited infinitely often is the set of all transitions of $\cal E$, with probability one~\cite[Lemma 10.119]{BaierK2008}.
  $\mathfrak{S}$ induces a scheduler $\mathfrak{S}_2$ on $\M$ from $s$ which satisfies: $\Pr_{\M,s}^{\mathfrak{S}_2}(\{\pi \mid L(\pi) \text{ is accepted from } q \text{ in } \mathcal{A}\}) = 1$.
  Combining schedulers $\mathfrak{S}_1$ and $\mathfrak{S}_2$ yields a scheduler witnessing $\prb^{\max}_{\M}(\lang(\mathcal{A})) > 0$.

  ``$\implies$'': Every path $\pi$ from $s$ in $\M$ induces a unique path $p_q(\pi)$ from $(s,q)$ in $\M \times \mathcal{A}$, if $q \in Q_D$ (assuming that $\mathcal{A}$ is complete).
  We let $\mdplimit(\pi)$ be the pair $(A,T)$ where $A$ is the set of states appearing infinitely often in path $\pi$ and $T : A \to 2^{\Act}$ is the set of actions appearing infinitely often for each of the states in $A$.
  Given a state $q \in Q_D$ and an end-component $\cal E$ of $\M \times \mathcal{A}$, we let 
  \begin{align*}
    X_{q,{\cal E}} = \{ \pi \mid &\text{ there exist } \pi_1\pi_2 \text{ s.t. } \pi = \pi_1 \pi_2, \\
    &I \xrightarrow{L(\pi_1)}_{\mathcal{A}} q \text{ and } \mdplimit(p_q(\pi_2)) = \mathcal{E}\}
  \end{align*}
  Let $\mathfrak{S}$ be a scheduler on $\M$ satisfying $\Pr^{\mathfrak{S}}_{\M}(\lang(\mathcal{A})) > 0$.
  The $\mathfrak{S}$-paths $\pi$ in $\M$ satisfying both $L(\pi) \in \lang(\mathcal{A})$ and $\pi \not\in \bigcup_{q \in Q_D} \{X_{q,\mathcal{E}} \mid \mathtt{atrans}(\mathcal{E}) \models \alpha\}$ form a null-set.
  This is because a path $\pi$ satisfies the following property with probability one under $\mathfrak{S}$: for all $i \geq 0$ and $q \in Q_D$: $\mdplimit(p_q(\pi[i..]))$ forms an end-component.
  It follows that there exists $q \in Q_D$ and an end-component $\cal E$ of $\M \times \mathcal{A}$ satisfying $\mathtt{atrans}(\mathcal{E}) \models \alpha$ such that $\Pr^{\mathfrak{S}}_{\M}(X_{q,\mathcal{E}}) > 0$.
  But then $\mathcal{E}$ must also be reachable in $\M \times \mathcal{A}$, which concludes the proof of the claim.
  
  It is a direct consequence now that the problem is in NP, as we can guess the end-component $\cal E$ and then check in polynomial time whether $\mathtt{atrans}({\cal E}) \models \alpha$ holds.

  \textbf{In P for fin-less limit-det. TELA.}
  If $\alpha$ is fin-less, then the existence of and end-component $\cal E$ satisfying $\mathtt{atrans}({\cal E}) \models \alpha$ is equivalent to the existence of a \emph{maximal} end-component satisfying the same property.
  This is because fin-less properties are preserved when adding additional transitions.
  As all maximal end-components can be enumerated in polynomial time (see Algorithm 47 in~\cite{BaierK2008}), it follows by the above claim that $\prb^{\max}_{\M}(\lang(\mathcal{A})) > 0$ can be decided in polynomial time in this case.
\end{proof}

Before giving a proof of \Cref{thm:GFMcorrectness} we prove a few lemmas related to the breakpoint construction.
Let $\mathcal{A} = (Q,q_0,\Sigma,\delta,\alpha)$ be a TELA in DNF, with $\alpha = \bigvee_{1 \leq i \leq m} \alpha_i$ and $\alpha_i = \finset(T^i_0) \land \bigwedge_{1 \leq j \leq k_i} \infset(T^i_j)$ (see~\Cref{def:dnf}) and let the breakpoint automata $\breakp_i$ be defined as in~\Cref{def:breakp}.
As defined in \Cref{sec:limitdet}, we let $\theta$ be the extended subset transition function corresponding to $\delta$.
Also, $\theta_i = \theta|_{\delta \setminus T_i^0}$ is defined as $\theta$ but with $\delta$ restricted to transitions outside of $T_i^0$.
We call a path a $\theta_i$-path if for all its transitions $(p,a,p')$ we have $p' \in \theta_i(\{p\},a)$.
For an infinite word $w = w_0 w_1 w_2 \ldots$ we let $w[j..m] = w_j w_{j+1} \ldots w_m$.

The proofs follow known arguments for the correctness of limit-determinization for B\"uchi and generalized B\"uchi automata (see~\cite[Section 4.2]{CourcoubetisY1995} and~\cite[Sections 7.4 and 7.6]{Blahoudek2018}).
Still, the extension to Emerson Lei requires additional arguments and hence we give the proofs here in our notation for completeness.
The first lemma extends~\cite[Lemma 7.1]{Blahoudek2018}.
\begin{lemma}
  \label{lem:breakphelper}
  For every accepting run $\rho = q_0q_1 \ldots$ of $\mathcal{A}$ for $w = w_0 w_1 \ldots$ there exists an $1 \leq i \leq m$ such that $\rho \models \alpha_i$, and a $K \geq 0$ such that:
  \begin{itemize}
  \item for all $l \geq K$: $(q_l,w_l,q_{l+1}) \notin T_0^i$,
  \item for all $l \geq K$ there exists $m > l$ such that $\theta_i(\{q_{l}\},w[l..m]) = \theta_i(\{q_{K}\},w[K..m])$.
  \end{itemize}
\end{lemma}
\begin{proof}
  As $\rho$ is accepting there must exist an $1 \leq i \leq m$ such that $\rho \models \alpha_i$, which implies $\rho \models \finset(T_0^i)$.
  Fix such an $i$.
  Then, the existence of a $K$ (let us call it $K_1$) satisfying the first condition follows directly.
  Fix such a $K_1$.

  Suppose, for contradiction, that for all $K_2 \geq K_1$ there exists $l_1 \geq K_2$ such that for all $m > l_1$ we have $\theta_i(\{q_{l_1}\},w[l_1..m]) \neq \theta_i(\{q_{K_2}\},w[K_2..m])$.
  Clearly then $\theta_i(\{q_{l_1}\},w[l_1..m]) \subset \theta_i(\{q_{K_2}\},w[K_2..m])$ as $q_{l_1} \in \theta_i(\{q_{K_2}\},w[K_2..{l_1{-}1}])$.
  Applying the same argument lets us find $l_2 \geq l_1$ such that for all $m > l_2$ we have $\theta_i(\{q_{l_2}\},w[l_2..m]) \subset \theta_i(\{q_{l_1}\},w[l_1..m])$.
  Iterating this argument lets us construct an infinitely descending chain, which is impossible as all considered sets are finite.

  It follows that there exists a $K_2 \geq K_1$ satisfying the second property, which concludes the proof.
\end{proof}

\begin{lemma}
  \label{lem:gfmlangsubset}
  $\lang(\gfm_{\mathcal{A}}) \subseteq \lang(\mathcal{A})$ and $\lang(\ld_{\mathcal{A}}) \subseteq \lang(\mathcal{A})$.
\end{lemma}
\begin{proof}
  We give the argument only for $\gfm_{\mathcal{A}}$, as it is the same for $\ld_{\mathcal{A}}$ with simple modifications.
  Any accepting run of $\gfm_{\mathcal{A}}$ for any $w_0 w_1 \ldots \in \Sigma^{\omega}$ has the form
  \[I \to Q_1 \to \ldots \to Q_k \to (P_0,\varnothing,0)_i \to (P_1,B_1,h_1)_i \to \ldots \]
  where the first part from $I$ to $Q_k$ is a path through the initial component of $\gfm_{\mathcal{A}}$, and the subsequent part is a path through some breakpoint component $\breakp_i$, with $1 \leq i \leq m$.
  As $\breakp_i$ is the standard breakpoint automaton for $\mathcal{A}$ under acceptance $\alpha_i$ and after removing transitions in $T_0^i$, it follows from the soundness of the known construction~\cite[Lemma 7.3]{Blahoudek2018} that there exists a run $\rho$ through $\mathcal{A}$ from some state $q \in P_0$ for the word $w[k{+}1..]$ such that $\rho \models \alpha_i$ and $\rho$ sees no transition in $T_0^i$.
  Furthermore, clearly any state in $P_0$ is reachable from $I$ with a path labeled by $w[0..k]$.
  Concatenating such a path with the run $\rho$ yields an accepting run of $\mathcal{A}$ for $w_0 w_1 \ldots$.
\end{proof}
The last part of the following proof is essentially the proof of~\cite[Lemma 7.2]{Blahoudek2018}.
\GFMcorrectness*
\begin{proof}
  As $\delta_i^{\operatorname{break}} \subseteq \delta_i$ holds for all $i$, and the breakpoint components are deterministic, both $\gfm_{\mathcal{A}}$ and $\ld_{\mathcal{A}}$ are syntactically limit-deterministic.
  The bound $O(2^n + 3^n \, m \, k)$ on the states of $\gfm_{\mathcal{A}}$ follows as the state-space of $\breakp_i$ is bounded by $3^n \cdot k$ (this uses that $B \subset R$ holds for states $(R,B,l)$ of $\breakp_i$, and $l \leq k$) for all $i$, and there are $m$ breakpoint components.
  Additionally, the initial component is a direct subset construction of $\mathcal{A}$, adding another $2^n$ states.
  For $\ld_{\mathcal{A}}$, the initial component is of size $n$, hence it has state-complexity $O(n + 3^n \, m \, k)$.

  It remains to show that $\lang(\ld_{\mathcal{A}}) = \lang(\gfm_{\mathcal{A}}) = \lang(\mathcal{A})$.
  The inclusions in $\lang(\mathcal{A})$ follow by~\Cref{lem:gfmlangsubset}.
  For the other direction, let $\rho = q_0q_1 \ldots$ be an accepting run of $\mathcal{A}$ for $w$.
  By~\Cref{lem:breakphelper} there exist $1 \leq i \leq m$ and $K \geq 0$ such that:
  \begin{itemize}
  \item $\rho \models \alpha_i$
  \item for all $l \geq K$: $(q_l,w_l,q_{l+1}) \notin T_0^i$,
  \item for all $l \geq K$ there exists $m > l$ such that $\theta_i(\{q_{l}\},w[l..m]) = \theta_i(\{q_{K}\},w[K..m])$.
  \end{itemize}
  We define a run of $\gfm_{\mathcal{A}}$ for $w$ as follows.
  For the first $K{-}1$ steps, it remains inside the initial component.
  It will reach a state $P \subseteq Q$ with $q_{K{-}1} \in P$.
  Then it takes the transition $\bigl(P,w[K{-}1],(\{q_{K}\},\varnothing,0)_i\bigr)$, and continues deterministically thereafter in component $\breakp_i$.
  This transition exists as $(q_{K{-}1},w[K{-}1],q_K) \in \delta$, and hence $\{q_K\} \subseteq \theta(P,w[K{-}1])$.
  We argue that the remaining run $(\{q_{K}\},\varnothing,0) (R_1,B_1,h_1) \ldots$ sees infinitely many transitions in $\delta_i^{\operatorname{break}}$.
  For contradiction, suppose that this is not the case.
  Then, for some $N \geq K$ we have: for all $j \geq N: \; B_j \subset R_j$ and $h_j = h_{j+1}$.
  Furthermore, there exists $N_1 > N$ such that $(q_{N_1},w[N_1],q_{{N_1}{+}1}) \in T_{i}^{h_j}$ as $\rho \models \alpha_i$.
  It follows that $q_{{N_1}{+}1} \in B_{{N_1}{+}1}$.
  By the third property above, however, there exists $m > N_{1}{+}1 \geq K$ such that $R_m = \theta_i(\{q_K\},w[K..m]) = \theta_i(\{q_{N_{1}{+}1}\},w[N_1{+}1..m]) \subseteq B_m$.
  This contradicts $B_m \subset R_m$.

  A run for $\ld_{\mathcal{A}}$ can be constructed in the same way, with the only exception that the initial part of the run is set to be $q_0 \ldots q_{K-1}$ (rather than the unique path in the subset-construction for prefix $w[0..{K{-}2}]$).
  In total, this shows $\lang({\cal A}) = \lang({\gfm_{\cal A}})$ and $\lang({\cal A}) = \lang({\ld_{\cal A}})$, which concludes the proof.
\end{proof}

Next, we turn to the proof that $\gfm_{\mathcal{A}}$ is good-for-MDPs.
First, we prove a lemma which gives a sufficient condition for a word to be accepted by $\breakp_i$.

\begin{lemma}
  \label{lem:suffcondbreak}
  Let $1 \leq i \leq m$, $w \in \Sigma^{\omega}$, $0 = j_0 < j_1 < j_2 < \ldots$ be an increasing sequence of natural numbers and $Q' \subseteq Q$ satisfying:
  \begin{itemize}
  \item for all $l \geq 0$: $\theta_i(Q', w[j_l..{j_{l{+}1}{-}1}]) = Q'$ and
  \item for all $l \geq 0$ and $q' \in Q'$ there exists $q \in Q'$ and a finite $\theta_i$-path from $q$ to $q'$ in $\mathcal{A}$ labeled by $w[j_l..j_{l+1}{-}1]$ which hits all transition-sets $T_1^i, \ldots, T_{k_i}^i$.
  \end{itemize}
Then, $w$ is accepted from state $(Q',\varnothing,0)$ in $\breakp_i$.
\end{lemma}
\begin{proof}
  As the first component of states in $\breakp_i$ just simulates the subset transition function $\theta_i$, the run from state $(Q',\varnothing,0)$ of $\breakp_i$ for $w$ has the form:
  \[(Q',\varnothing,0) \xrightarrow{w[0..j_1{-}1]} (Q',B_1,h_1) \xrightarrow{w[j_1..j_2{-}1]} (Q',B_2,h_2) \xrightarrow{w[j_2..j_3{-}1]} \ldots\]
  We show that each infix $(Q',B_l,h_l) \xrightarrow{w[j_l..j_{l{+}1}{-}1]} (Q',B_{l{+}1},h_{l{+}1})$ of the run above sees a transition in $\delta_i^{\operatorname{break}}$.
  For contradiction, suppose that for some $l$ this is not the case.
  Then the third component is not changed in any transition of the corresponding infix run. 
  Let
  \[(Q_0,A_0,h_l) \xrightarrow{w[j_l]} (Q_1,A_1,h_l) \xrightarrow{w[j_l{+}1]} \ldots \xrightarrow{w[j_l{+}M{-}1]} (Q_M,A_M,h_l)\]
  be the sequence of the run between $(Q_0,A_0,h_l) = (Q',B_l,h_l)$ and $(Q_M,A_M,h_l) = (Q',B_{l+1},h_l)$.
  Take arbitrary $q' \in Q' \setminus B_{l+1}$ and let $P$ be the source states of the incoming $w[j_l{+}M{-}1]$-transitions of $q'$ wrt. $\delta \setminus T_0^i$.
  It follows that $P \cap A_{M{-}1}$ is empty and for no $p \in Q_{M{-}1}$ we have $(p,w[j_l{+}M{-}1],q') \in T_{h_l}^i$, as otherwise we would either have $q' \in B_{l+1}$ or see a break-point transition.
  In particular, there is no path from any $p \in Q_{M{-}1}$ to $q'$ labeled by $w[j_l{+}M{-}1]$ that sees $T_{h_l}^i$.
  We can continue this argument inductively until reaching $(Q_0,A_0,h_l)$ and conclude that there is no $q \in Q'$ such that there exists a $\theta_i$-path from $q$ to $q'$ labeled by $w[j_l..j_{l+1}{-}1]$ which sees some transition in $T_{h_l}^i$.
  This is in contradiction with the second property above.

  Hence the run sees a transition in $\delta_i^{\operatorname{break}}$ infinitely often, which implies that it is accepting.
\end{proof}

Recall that $\removeFin(\splitTELAi{\mathcal{A}}{i})$ is a GBA with two components: the initial component $Q_1$ is a copy of $\mathcal{A}$ where no transition is accepting, and component $Q_2$ is a copy of $\mathcal{A}$ without the transitions of $T_0^i$. The acceptance condition is the GBA condition $\alpha_i = \bigwedge_{1 \leq j \leq k_i} \infset(U_{j}^i)$ where $U_j^i$ is the set of transitions corresponding to $T_j^i \setminus T_0^i$ in the component $Q_2$.

Now let $\mathcal{D}_i = \determin(\removeFin(\splitTELAi{\mathcal{A}}{i}))$ be the Rabin automaton one gets by applying the construction of~\cite{ScheweV2012} to $\removeFin(\splitTELAi{\mathcal{A}}{i})$.
Their soundness-proof\cite[Thm. 1]{ScheweV2012} can be modified slightly to prove the following statement.
Let $\gamma$ be the extended subset transition function of $\removeFin(\splitTELAi{\mathcal{A}}{i})$.
\begin{lemma}
  \label{lem:connectiontodet}
  Let $d_0 d_1 \ldots$ be an accepting run of $\mathcal{D}_i$ for $w$ satisfying Rabin pair $\infset(D_1) \land \finset(D_2)$ and $i_0$ be a position such that for all $l \geq i_0$ we have $(d_l,w[l],d_{l{+}1}) \notin D_2$ and for infinitely many $l \geq i_0$ we have $d_{l} = d_{i_0}$.

  Then, there exists a sequence $i_0 < i_1 < i_2 \ldots$ and a $Q' \subseteq Q_2$ such that
  \begin{itemize}
  \item $Q' \subseteq \gamma(I,w[0..i_0{-}1])$,
  \item for all $l \geq 0$: $Q' = \gamma(Q',w[i_l..i_{l+1}{-}1])$ and for all $q_1 \in Q'$ there exists a $q_2 \in Q'$ and a finite $\gamma$-path for $w[i_l..i_{l+1}{-}1]$ hitting all transition-sets $U_1^i, \ldots, U_{k_i}^i$.
  \end{itemize}
\end{lemma}
\begin{proof}
  The following statement follows directly from the proof of\cite[Thm. 1]{ScheweV2012}.
  Let $j_0 < j_1 < \ldots$ be an infinite sequence of positions such that: for $l \geq j_0$ we have $(d_l,w[l],d_{l{+}1}) \notin D_2$ and the sequence $d_{j_l}w[j_l]d_{j_l{+}1} \ldots d_{j_{l{+}1}}$ sees $k_i$ transitions in $D_1$ for all $l\geq 0$.
  Then there exists a sequence $P_0, P_1 \ldots$, with all $P_l \subseteq Q_1 \cup Q_2$, such that $P_0 \subseteq \gamma(I,w[0..j_0{-}1])$, $P_{l{+}1} = \gamma(P_l,w[j_l..j_{l{+}1}{-}1])$ and for all $q_1 \in P_{l{+}1}$ there exists a $\gamma$-path for $w[j_l..j_{l{+}1}{-}1]$ hitting all transition sets $U_1^i,\ldots, U_{k_i}^i$.
  Furthermore, the sequence is such that if $d_{j_l} = d_{j_p}$, then $P_l = P_p$.

  Clearly a sequence $j_0 < j_1 < \ldots$ exists satisfying the above assumption (as $\rho$ is accepting) and which additionally has the property that $d_{j_l} = d_{j_0}$ for all $l \geq 0$ (by taking $j_0 = i_0$ and the fact that infinitely many positions see $d_{i_0}$).
  We may conclude that the sequence $P_0,P_1 \ldots$ exists as above and as $d_{j_l} = d_{j_0}$ for all $l \geq 0$, indeed $P_0 = P_l$ for all $l$.
  So we may take $Q' = P_0$, and the only thing left to show is that $Q' \subseteq Q_2$ holds.
  This, however, follows from the fact that $Q'$ is reachable from some accepting transition.
\end{proof}
We are now in a position to prove~\Cref{lem:goodimpliesacc}.
Recall that $\mathcal{D} = \bigotimes_{1 \leq i \leq m} \mathcal{D}_i$ and $\mathfrak{S}$ is a finite memory scheduler on $\M$.
\Goodimpliesacc*
\begin{proof}
  There exists some Rabin pair $\infset(D'_1) \land \finset(D'_2)$ of $\cal D$ such that no transition of $B$ is contained in $D_2'$, and some transition of $B$ is contained in $D_1'$.
  By construction of $\mathcal{D}$, this pair directly corresponds to a Rabin pair $\infset(D_1) \land \finset(D_2)$ of one of the components $\mathcal{D}_i$ (as the acceptance condition of $\mathcal{D}$ is essentially the disjunction of the acceptance conditions of the $\mathcal{D}_i$).
  
  Hence for every path $\pi_2$ through $B$ starting in $\mathfrak{s}$ which sees all transitions in $B$ infinitely often we find an accepting run $d_0 d_1 \ldots $ of $\mathcal{D}_i$ for $L(\pi_1 \pi_2)$ and a sequence $j_0 < j_1 < \ldots$ where $j_0 = |\pi_1|$ and such that:
  \begin{itemize}
  \item no transition of $d_{j_0} d_{j_0{+}1} \ldots$ is included in $D_2$,
  \item for all $l \geq 0$: $d_{j_l} = d_{j_0}$ and the run $d_0 d_1 \ldots$ sees an accepting transition in between position $j_l$ and $j_{l+1}$ for all $l \geq 0$.
  \end{itemize}
  By~\Cref{lem:connectiontodet} there exists a $Q' \subseteq \gamma(I,L(\pi_1)) \cap Q_2$ such that for all $l \geq 0$: $Q' = \gamma(Q',w[j_l..j_{l+1}{-}1])$ and for all $q_1 \in Q'$ there exists a $q_2 \in Q'$ and a finite $\gamma$-path for $w[j_l..j_{l+1}{-}1]$ hitting all transition-sets $U_1^i, \ldots, U_{k_i}^i$.
  Let $Q''$ be the corresponding set of states of $\mathcal{A}$, which implies $Q'' \subseteq \theta(I,L(\pi_1))$.
  As $Q' \subseteq Q_2$ we have: for all $q_1 \in Q''$ there exists a $q_2 \in Q''$ and a finite $\theta_i$-path for $w[j_l..j_{l+1}{-}1]$ hitting all transition-sets $T_1^i, \ldots, T_{k_i}^i$.

  It follows that the conditions of~\Cref{lem:suffcondbreak} are satisfied and hence that $L(\pi_2)$ is accepted from $(Q'',\varnothing,0)$ in $\breakp_i$.
\end{proof}

\nontrivinclusionGFM*
\begin{proof}
  Let $\mathfrak{S}$ be a finite-memory scheduler on $\M$ and $\M_{\mathfrak{S}} \times \mathcal{D}$ be as above.
  We construct the scheduler $\mathfrak{S}'$ as follows.
  For every finite path $\pi_1$ of $\M \times \gfm_{\mathcal{A}}$ it needs to choose an action of $\M$ and whether to move from the initial component to one of the breakpoint components (this is the only nondeterministic choice in $\gfm_{\mathcal{A}}$).
  The action of $\M$ is always chosen in the same way as by $\mathfrak{S}$ for the corresponding finite path of $\M$.
  If the path in $\M_{\mathfrak{S}} \times \mathcal{D}$ corresponding to $\pi_1$ does not end in a BSCC, $\mathfrak{S}'$ chooses to remain in the initial component.
  If it does reach a BSCC $B$ we make a distinction on whether $B$ is accepting or not.
  If $B$ is not accepting, then $\mathfrak{S}'$ can be defined arbitrarily for all prefixes that extend $\pi_1$.
  If $B$ is accepting, then let us assume that $\last(\pi_1) = (s, P)$ is the current state of $\M \times \gfm_{\mathcal{A}}$.
  Then $P$ is exactly the set of states reachable in $\mathcal{A}$ from $I$ on a path labeled by $L^{\times}(\pi_1)$.
  By~\Cref{lem:goodimpliesacc} there exists $1 \leq i \leq k$ and $Q' \subseteq P$ such that the probability of generating a suffix $\pi_2$ from state $s$ in $\M$ under scheduler $\mathfrak{S}$ whose label is accepted in $\breakp_i$ from state $(Q',\varnothing,0)$ is $1$.
  Consequently, $\mathfrak{S'}$ chooses $\bigl(\theta(Q',L(s)),\varnothing,0\bigr)_i$ as successor state of the automaton and continues to simulate $\mathfrak{S}$.
  As the probability of generating a path $\pi$ such that $L(\pi) \in \lang(\mathcal{A})$ in $\M_{\mathfrak{S}}$ is equivalent to the probability of reaching an accepting BSCC in $\M_{\mathfrak{S}} \times \mathcal{D}$, we can conclude that $\Pr_{\M \times \gfm_{\mathcal{A}}}^{\mathfrak{S}'}(\Pi_{acc}) \geq \Pr_{\M}^{\mathfrak{S}}(\lang(\mathcal{A}))$.
\end{proof}

\GFMproperty*
\begin{proof}
  To show that $\gfm_{\mathcal{A}}$ is GFM, it suffices to show the following two statements:
  \begin{enumerate}
  \item For every finite memory scheduler $\mathfrak{S}$ of $\M$ there exists a scheduler $\mathfrak{S}'$ of $\M \times \gfm_{\mathcal{A}}$ such that
    \[\Pr^{\mathfrak{S}}_{\M}(\lang(\mathcal{A})) \leq \Pr^{\mathfrak{S}'}_{\M \times \gfm_{\mathcal{A}}}(\Pi_{acc})\]
  \item For every finite memory scheduler $\mathfrak{S}$ of $\M \times \gfm_{\mathcal{A}}$ there exists a scheduler $\mathfrak{S}'$ of $\M$ such that
    \[\Pr^{\mathfrak{S}}_{\M \times \gfm_{\mathcal{A}}}(\Pi_{acc}) \leq \Pr^{\mathfrak{S}'}_{\M}(\lang(\mathcal{A}))\]
  \end{enumerate}
  1. is proved by~\Cref{lem:nontrivinclusionGFM} and 2. is proved for arbitrary automaton in\cite[Thm. 1]{KleinMBK2014}.
\end{proof}

\FloatBarrier

\section{Additional evaluation of experiments}

\begin{figure}[tbp]
	\centering
	\begin{subfigure}{0.5\textwidth}
			\resizebox{\textwidth}{!}{
			\includegraphics{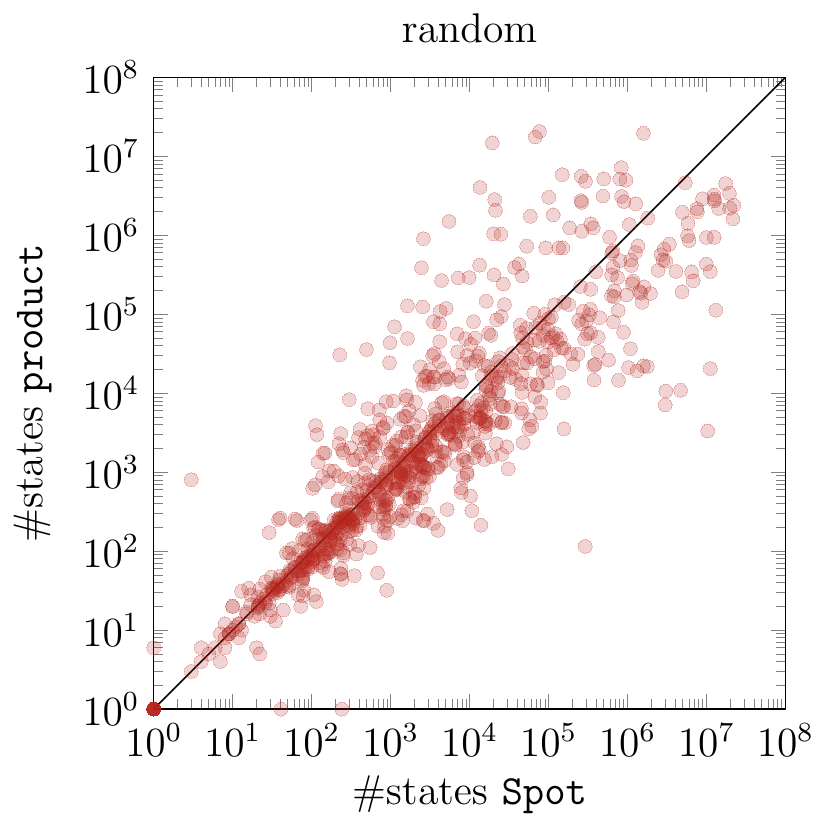}}
		\end{subfigure}%
		\begin{subfigure}{0.49\textwidth}
			\resizebox{\textwidth}{!}{
				\includegraphics{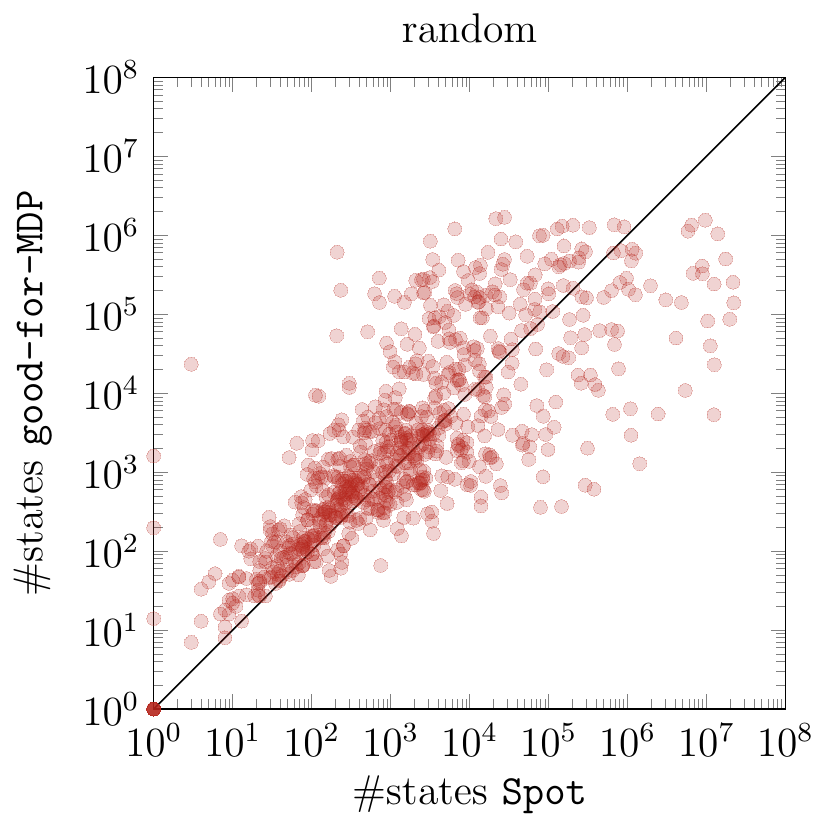}}
		\end{subfigure}
		\begin{subfigure}{0.5\textwidth}
			\resizebox{\textwidth}{!}{
				\includegraphics{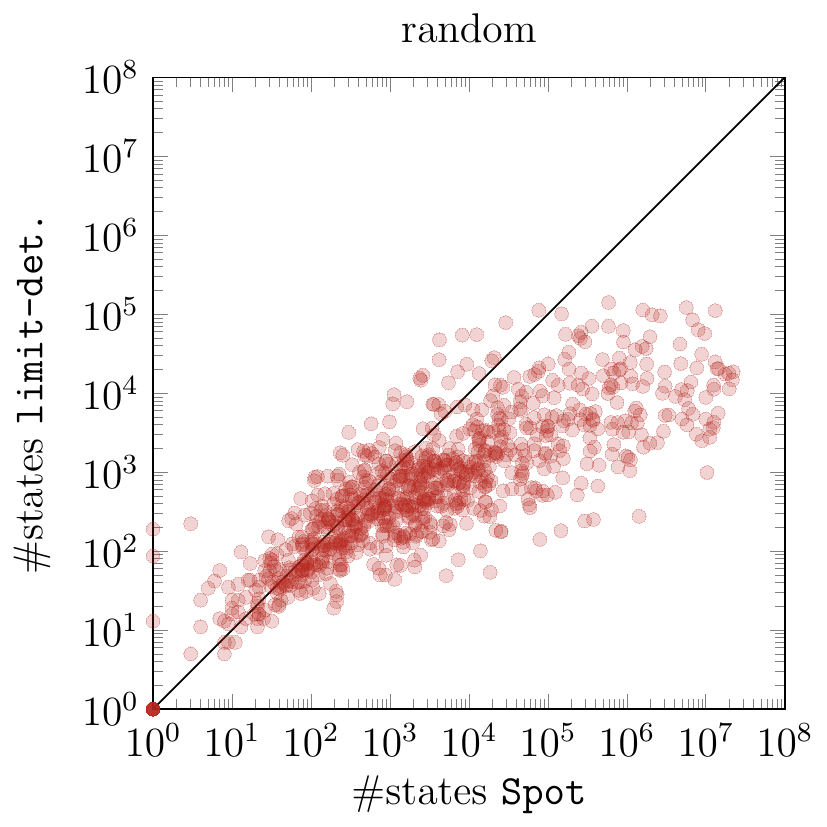}}
		\end{subfigure}%
		\begin{subfigure}{0.49\textwidth}
			\resizebox{\textwidth}{!}{
				\includegraphics{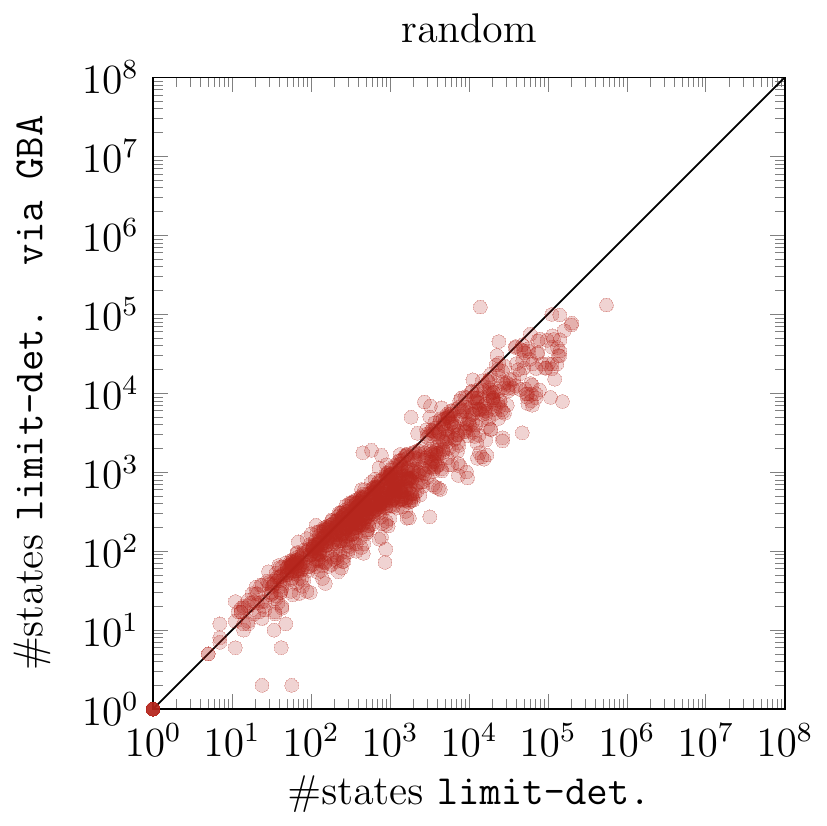}}
		\end{subfigure}
		\begin{subfigure}{0.5\textwidth}
			\resizebox{\textwidth}{!}{
				\includegraphics{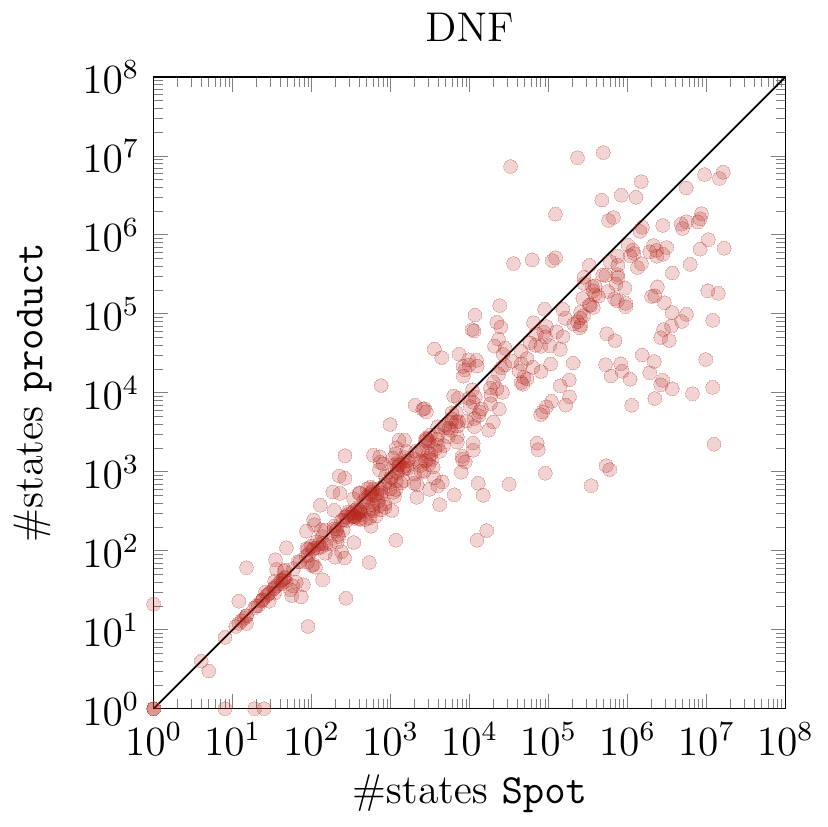}}
		\end{subfigure}
	\caption{Comparing pairwise the sizes of the state spaces for different approaches on the \emph{random} benchmark and \emph{DNF} benchmark.}
	\label{fig:EvaluationPairwise}
\end{figure}
\Cref{fig:EvaluationPairwise} compares the approaches pairwise on the \emph{random} benchmark and on the \emph{DNF} benchmark. 
On the \emph{random} benchmark, the comparisons of \texttt{product} and \texttt{good-for-MDP} with \texttt{\spot} show that the results are widely spread, i.e. there is no clear correlation. 
Both approaches produce in many cases smaller automata than \texttt{\spot}.
Most automata produced by \texttt{limit-det.} are smaller than the results produced by \texttt{\spot}. 
Such behavior is expected as the complexity is only single-exponential instead of double-exponential.
The approach \texttt{limit-det. via GBA} produces in almost all cases smaller results than the approach \texttt{limit-det.}
On the \emph{DNF} benchmark \texttt{product} produces in most cases smaller results than \texttt{\spot}. 

\begin{table}[tbp]
	\caption{Comparison of \texttt{\spot} and \texttt{product}, with input automata grouped by the size of the DNF of their acceptance condition and the amount of nondeterminism. ``states'', ``time'' and ``acceptance'' refer to the median of the ratio \texttt{product} / \texttt{\spot{}}. The number of automata that is denoted in brackets is the number of input automata for which both approaches were able to construct a result within the given time and memory bounds.}
	\label{tab:BenchmarksGroupedNondeterminism}
	\begin{center}
		\resizebox{\linewidth}{!}{
			\setlength{\tabcolsep}{4pt}
			\begin{tabular}{l l l l l l l l l l l l }
				\multirow{2}{*}{benchmark} & amount of & input & \#automata & \multicolumn{2}{c}{timeouts}& \multicolumn{2}{c}{memouts} & \multirow{2}{*}{states}  & \multirow{2}{*}{time} & \multirow{2}{*}{acceptance}  \\ 
				& nondet. & acceptance & (no time-/memouts) & \text{\spot{}} & \texttt{product} &\texttt{\spot{}} & \texttt{product} &   \\
				\hline
				\multirow{6}{*}{random}& \multirow{2}{*}{$\leq  0.66$} & $2 \leq |\alpha| \leq 11$ & 175 (175) & 0.0\% & 0.0\% & 0.0\% & 0.0\% & 0.91 & 1.37 & 1.80 \\
				& & $12 \leq \length{\alpha} \leq 21$ & 180 (180) & 0.0\% & 0.0\% & 0.0\% & 0.0\% & 0.82 & 1.44 & 1.64 \\ \cline{2-11}
				& \multirow{2}{*}{$> 0.66, \leq 1.33$} & $2 \leq |\alpha| \leq 11$ & 162 (159) & 0.0\% & 0.0\% & 0.0\%  & 1.9\% & 0.85 & 1.19 & 1.80\\
				& & $12 \leq \length{\alpha} \leq 21$ & 159 (144) & 0.0\% & 0.6\% & 3.1\% & 7.5\% & 0.75 & 1.26 & 1.82 \\  \cline{2-11}
				& \multirow{2}{*}{$ > 1.33$} & $2 \leq |\alpha| \leq 11$ & 163 (128) & 1.8\% & 0.0\% & 16.0\% & 11.7\% & 0.75 &  0.84 & 1.72\\
				& & $12 \leq \length{\alpha} \leq 21$ & 161 (82) & 1.2\% & 7.5\% & 42.2\% & 28.0\% & 0.43 & 0.80 & 1.67\\
				\hline \hline
				\multirow{3}{*}{DNF} & $\leq  0.66$ & $2 \leq |\alpha| \leq 21$ & 193 (193) & 0.0\% & 0.0\% & 0.0\% & 0.0\% & 0.81 & 1.17 & 1.14\\
				& $> 0.66, \leq 1.33$ & $2 \leq |\alpha| \leq 21$ & 167 (160) & 0.0\% & 0.0\% & 3.0\% & 2.4\% & 0.68 & 0.93 & 1.19\\
				& $ > 1.33$ & $2 \leq |\alpha| \leq 21$ & 140 (111) & 1.4\% & 0.0\% & 18.6\% & 10.7\% & 0.26 & 0.25 & 1.00
		\end{tabular}}
	\end{center}
\end{table}
\Cref{tab:BenchmarksGroupedNondeterminism} compares the approaches \texttt{\spot} and \texttt{product} in more detail. The input automata are grouped by the amount of nondeterminism and the length of the input acceptance condition, which is in DNF. These groups are the same as in \Cref{fig:SpotVsProduct}. The values ``states'', ``time'' and ``acceptance'' are the median value of the ratio of these metrics, which are the values depicted in \Cref{fig:SpotVsProduct}. 
If only one of the approaches was able to construct a result within the given time and memory bounds we define the ratio as 0 (\texttt{\spot} did not construct a result) or infinity (\texttt{product} did not construct a result). If none of the approaches was able to construct a result, we do not consider this input automaton for the calculation.
The number of timeouts and memouts increases with the amount of nondeterminism and the size of the input acceptance condition.
\texttt{product} produces for all subsets automata with fewer states than \texttt{\spot{}} and the ratio of the computation time decreases the larger the amount of nondeterminism is.
On the \emph{DNF} benchmark, \texttt{product} produces significantly smaller automata while the acceptance condition is only slightly larger.  

\begin{figure}[tbp]
	\centering
		\includegraphics{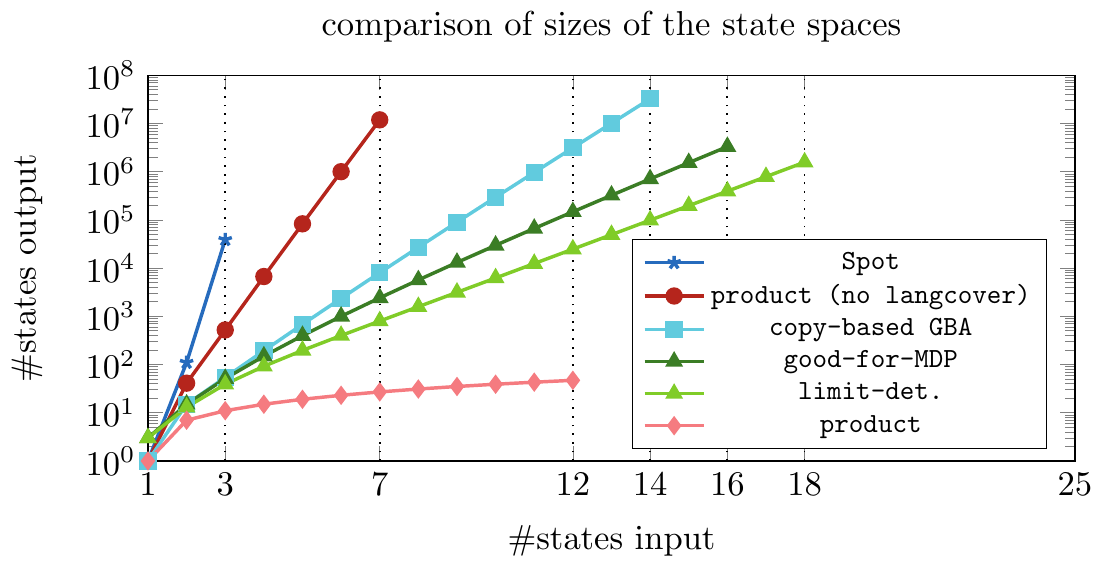}%
	\caption{Comparing the approaches for input automata of~\Cref{fig:ExpCNF}. The copy-based GBA approaches produce the same results and are collapsed here.}
	\label{fig:EvaluationExpCNF}
\end{figure}
\Cref{fig:EvaluationExpCNF} compares the approaches on the sequence of automata described in~\Cref{fig:ExpCNF}.
For these automata, \spot{} first computes GBA with an exponential number of acceptance sets.
The example also highlights the effect of the langcover heuristic: it benefits from the fact that the language of the input automaton under any individual disjunct of the DNF is the same.
This is used to prune large parts of the state-space.

\end{document}